\documentclass[a4paper, 12pt]{article}

\usepackage[english]{babel}

\usepackage[letterpaper,top=2cm,bottom=2cm,left=3cm,right=3cm,marginparwidth=1.75cm]{geometry}

\usepackage{setspace}
\usepackage{enumerate}
\usepackage{amsmath}{\large }
\usepackage{amssymb}
\usepackage{amsthm}
\usepackage{dsfont}
\usepackage{comment}
\usepackage[colorlinks=true, allcolors=blue]{hyperref}
\newtheorem{theorem}{Theorem}[section]
\newtheorem{lemma}{Lemma}[section]
\newtheorem{proposition}{Proposition}[section]
\newtheorem{definition}{Definition}[section]

\newtheorem{assumption}{Assumption}[section]

\newtheorem{example}{Example}[section]
\usepackage{actuarialsymbol}
\numberwithin{equation}{section}
\usepackage{enumerate}
\usepackage{enumitem}
\usepackage{yfonts}
\usepackage{mathtools}
\usepackage{chngcntr}
\counterwithout{figure}{section}
\DeclareMathOperator*{\argsup}{arg\,sup}

\usepackage[authoryear]{natbib}

\DeclareMathOperator*{\argmax}{argmax}
\usepackage{bbm}
\usepackage{booktabs}

\title{Nonconcave Portfolio Choice under Smooth Ambiguity}

\author{
Emanuele Borgonovo\thanks{Department of Decision Sciences, Bocconi University, 20136 Milan, Italy, emails: \texttt{emanuele.borgonovo@unibocconi.it}; \texttt{massimo.marinacci@unibocconi.it}.}
\quad
An Chen\thanks{Institute of Insurance Science, Ulm University, Helmholtzstr.\ 20, 89069 Ulm, Germany, emails: \texttt{an.chen@uni-ulm.de}; \texttt{shihao.zhu@uni-ulm.de}.}
\quad
Massimo Marinacci\footnotemark[1]
\quad
Shihao Zhu\footnotemark[2]
}

\begin{document}	
	\maketitle
	

	
	\begin{abstract}
	We study continuous-time portfolio choice with nonlinear payoffs under smooth ambiguity and Bayesian learning. We develop a general framework for dynamic, non-concave asset allocation that accommodates nonlinear payoffs, broad utility classes, and flexible ambiguity attitudes. Dynamic consistency is obtained by a robust representation that recasts the ambiguity-averse problem as ambiguity-neutral with distorted priors. This structure delivers explicit trading rules by combining nonlinear filtering with the martingale approach and nests standard concave and linear-payoff benchmarks. As a leading application, delegated management with convex incentives illustrates that ambiguity aversion shifts beliefs toward adverse states, limits the range of states that would otherwise trigger more aggressive risk taking, and reduces volatility through lower risky exposure.
	\end{abstract}

		\vspace{2mm}


\vspace{2mm}

\vspace{2mm}

\vspace{2mm}

\clearpage
\section{Introduction}\label{sec:Intro}

Decision-makers face uncertainty not only about outcomes (risk), but also about the probabilities governing those outcomes (ambiguity). Preferences therefore reflect attitudes toward both risk and ambiguity, with substantial evidence supporting ambiguity aversion \citep{marinacci2015model}. The recent literature on financial decision-making has intensively considered how to model these attitudes in both prescriptive and descriptive modes, as the works of \cite{balter2021consistency, guan2025kmm, guan2025equilibrium, bauerle2024optimal} illustrate. One of the challenges in developing a comprehensive modeling approach is accounting for not only risk and ambiguity preferences but also nonlinear payoffs arising from alternative investment forms and financial compensation mechanisms. A canonical instance is option-like terminal payoffs, as discussed in \cite{carpenter2000does} and \cite{dai2022nonconcave}. Moreover, the model must reflect the fact that investors revise their beliefs after receiving information, for instance, from market prices, thereby yielding a dynamic class of decision-making problems.


Prior work has typically examined these forces separately. \citet{carpenter2000does} studies option-type payoffs, but abstracts from both ambiguity and learning. \citet{bauerle2024optimal}, by contrast, incorporates ambiguity aversion and Bayesian learning, yet confines attention to linear payoffs and therefore misses the convex incentives central to option-based compensation. We bring these elements together in a general continuous-time portfolio framework with nonlinear payoffs, ambiguity aversion, and Bayesian learning. This framework yields tractable optimal trading rules and sharp comparative statics, providing a systematic characterization of how incentive convexity, belief updating, and ambiguity interact in dynamic financial decision-making.



We aim to theoretically examine the effects of ambiguity preferences, analyzing the whole spectrum from near-ambiguity neutrality to complete aversion (close to min-max decision-making). The ideal ambiguity representation is then that of the smooth-ambiguity functional of \cite{klibanoff2005smooth} (KMM, henceforth). The KMM functional separates beliefs from tastes toward ambiguity, uses a concave aggregator to capture dislike for ambiguous bets, and yields a smooth value index that eases calibration, estimation, and comparative statics. 

However, when combining smooth ambiguity with non-concavity objectives, two key challenges emerge. First, the concavity of the ambiguity aggregator breaks dynamic consistency (see \cite{hanany2009updating, savochkin2025dynamic}), thus impairing the inclusion of belief updates. Second, nonconcave maximization problems induced by option-style payoffs invalidate off-the-shelf convex optimization tools.  We resolve these challenges as follows.
We resort to a representation of quasiconcave functionals that transforms the ambiguity-averse problem into a Bayesian adaptive problem with endogenously distorted priors \citep{cerreia2011uncertainty, cerreia2011complete, drapeau2013risk,mazzon2024optimal}. We use this reformulation because, as shown in \cite{mazzon2024optimal}, it breaks the optimization problem into two parts. We propose a general result for our min-max formulation, where the (external) minimization is taken over the set of all possible priors and the (internal) maximization is equivalent to an ambiguity-neutral problem. The nested structure restores dynamic consistency because the maximization part is now a Bayesian expected utility problem. This makes the approach different from techniques that enforce time consistency by restricting the intertemporal structure of priors (often called ``rectangularity''; see \cite{savochkin2025dynamic}) or from game-theoretic resolutions that accept time inconsistency under concave utility (e.g., \cite{guan2025equilibrium}).

We then apply the resulting two-step approach to solve continuous-time portfolio problems with ambiguity preferences, learning, and non-concave objectives. In particular, we focus on financial decision problems with nonlinear payoffs. By applying nonlinear filtering, we update prior beliefs dynamically as market prices are observed. The internal (ambiguity-neutral) optimization problem is rewritten as an adaptive Bayesian problem with unknown drift (see, e.g.,\ \cite{karatzas2001bayesian, rieder2005portfolio}). Next, we combine the martingale method (e.g.,\ \cite[Ch.~3]{karatzas1998methods}) with concavification, replacing the original non-concave objective by its concave envelope to recover a tractable concave optimization problem. This yields explicit continuous-time trading rules for a broad class of genuinely non-concave problems. 

We complete the analysis by incorporating ambiguity preferences and performing an optimization over the set of priors. The approach yields an ambiguity‑neutral portfolio choice evaluated under an endogenously distorted prior that directly determines the optimal trading strategy, while accounting for ambiguity attitudes. We study the impact of ambiguity, deriving closed‑form characterizations for decision-makers whose preferences for risk and ambiguity are modeled, respectively, by power-power and power-exponential specifications. 

We then study how ambiguity shapes beliefs and investment policies. We analyze the impact of alternative risk preferences, ambiguity attitudes, and prior beliefs on the distorted prior, the optimal trading strategy, and terminal wealth, respectively.
The analysis shows that ambiguity aversion reallocates probability mass toward adverse states, effectively tightening risk constraints. As a result, optimal exposure to risky assets decreases in regions where the objective is locally convex in wealth or exposure. This pessimistic distortion disciplines incentives that would otherwise promote excessive risk-taking. By contrast, when the objective is locally concave, the impact of ambiguity is attenuated and aligns with standard robust‑investment prescriptions.

We illustrate these mechanisms in a leading application to delegated portfolio management with convex incentive contracts (see \cite{carpenter2000does}). In this setting, the endogenous prior distortion and the associated filtering dynamics translate into explicit portfolio rules and sharp quantitative predictions for risk taking. 

The remainder of the paper is organized as follows. Section ~\ref{sec:literature} overviews the related literature. Section ~\ref{sec_problem} formulates the problem. Section ~\ref{sec_SAP} sets up the decision functional, introducing the nested objective function. Section ~\ref{sec4} derives results for the problem of convex incentive contracts under smooth ambiguity preferences. 
Section~\ref{sec:numerical} presents detailed numerical experiments and comparative statics. Section ~\ref{sec_conclusion} offers conclusions.

\section{Related Literature}\label{sec:literature}
Our paper intersects previous literature on smooth ambiguity with learning in continuous time, robust representations for ambiguity and quasiconcave functionals, portfolio choice with non-concave objectives, and delegated portfolio management with convex incentives. Each of these research streams is broad in itself, and a comprehensive review is out of reach in the present work. We therefore provide a concise review of the aspects more salient for our work.


Regarding preferences in continuous time with learning, a growing body of work is studying asset allocation under smooth ambiguity.  \citet{balter2021consistency} study portfolio choice under smooth ambiguity, but restrict attention to deterministic strategies and therefore abstract from the subtle issue of time inconsistency. \citet{guan2025kmm} examine portfolio choice in incomplete markets from a pre-commitment perspective, in which the decision maker fixes a strategy at the initial date and does not revise it over time. Relative to this literature, we focus on time-consistent portfolio choice with learning. Our goal is to jointly accommodate Bayesian learning, genuinely nonlinear (option-style) payoffs, and broad utility/aggregator classes, while preserving dynamic consistency through a robust representation that delivers explicit trading rules and comparative statics. Closer to our setting, \citet{bauerle2024optimal} and \citet{guan2025equilibrium} allow for Bayesian learning under smooth ambiguity. \citet{guan2025equilibrium} incorporates Bayesian learning but maintains concave utility and derives equilibrium characterizations in a time-inconsistent environment. \citet{bauerle2024optimal}, by contrast, considers power utility and linear payoffs, and uses a dual representation based on \(L^p\) norms to obtain time-consistent solutions under smooth ambiguity. Our paper departs from these studies by allowing for general utility and aggregator classes in a dynamically consistent framework that also accommodates learning and genuinely nonlinear, option-like payoffs.

The goal of a general ambiguity representation connects us to the broad literature on dynamic choice under ambiguity. A key conceptual foundation is the robust representation of quasiconcave preference functionals, which provides a rationale for distorted priors; see \cite{cerreia2011uncertainty, cerreia2011complete, cerreia2011risk, drapeau2013risk}. In particular, when the smooth ambiguity problem is rewritten in certainty-equivalent form, Theorem~7 of \cite{cerreia2011uncertainty} and Proposition~6 of \cite{drapeau2013risk} deliver a sup--inf representation that is well suited to dynamic updating. \cite{mazzon2024optimal} exploits this type of reformulation in an optimal stopping problem.  Our paper is further related to \citet{hansen2018aversion,hansen2022asset}, who develop recursive representations of recursive smooth ambiguity preferences (\citet{klibanoff2009recursive}).  They show that ambiguity aversion can be recast as a robust distortion of priors and posteriors over the hidden Markov state, and study its implications for asset allocation and asset pricing.

When coming to financial-decision making, depending on the instrument at hand, problems can give rise to formulations with nonconcave utility maximization and option-style payoffs.
Examples include convex incentive contracts in \cite{carpenter2000does}, performance-fee schemes in \cite{he2018profit}, equity-linked insurance with guarantees in \cite{chen2019constrained}, and goal-based preferences in \cite{capponi2024continuous}. We refer to \cite{dai2022nonconcave} for a general treatment of non-concave portfolio maximization under portfolio constraints. Departing from this literature, which assumes known return dynamics, we introduce ambiguity in expected returns, incorporate Bayesian updating, and provide explicit decision rules using robust representation, filtering, and concavification.

A vast literature studies how contract convexity shapes risk-taking and induces lock-in effects; see \cite{carpenter2000does, basak2007optimal}. We revisit this setting through the lens of ambiguity and learning. 


Taken together, these research strands help us highlight the open issues addressed in this work, namely, to deliver a time-consistent solution method for continuous-time smooth-ambiguity portfolio problems with learning and non-concave payoffs, provide explicit trading policies and comparative statics, and offer portable insights for delegated management and embedded-option applications.


\section{Problem Formulation}\label{sec_problem}

Let $(\Omega,\mathcal{F},\mathbb{P})$  be  a  probability space equipped with a filtration $\mathbb{F}=(\mathcal{F}_t)_{t\in[0,T]}$  satisfying the usual conditions of completeness and right-continuity. Here $T>0$ is the fixed planning horizon, and the filtration $\mathbb{F}$ gathers all information available in the financial market.

Before modeling the investor's decisions under smooth ambiguity preferences, we introduce the source of ambiguity---uncertainty about the expected return---in the financial market setup. Assume that there are two assets that an investor can trade without any transaction costs: a risk-free and a risky asset. The risk-free asset (bond) has price process $S^0=(S^0_t)_{t\in[0,T]}$ governed by
\[
	dS^0_t = r S^0_t  dt,\qquad S^0_0 = s^0>0,
\]
where $r>0$ is the constant interest rate. The risky asset (stock) has the price process $S=(S_t)_{t\in[0,T]}$ satisfying 
\begin{equation}\label{stock}
	dS_t = Z S_t \, dt + \sigma S_t \, dB_t,\qquad S_0 = s>0,
\end{equation}
where $B=(B_t)_{t\in[0,T]}$ is an $\mathbb{F}$-adapted standard Brownian motion under $\mathbb P$, {$\sigma>0$ is the known constant volatility, and $Z\in\mathbb{R}$ denotes the expected return of the stock. As explained in \cite[footnote~2]{kardaras2012robust}, high-frequency data readily give good estimators for $\sigma$, while estimating $Z$ is much more statistically challenging. Therefore, we assume the investor faces uncertainty about the constant value of $Z \in \mathcal S$ for a compact metric space $\mathcal S \subset \mathbb R$. As time proceeds, the investor learns about the parameter $Z$ through observation of the stock prices $S$.}

 Assume that $Z$ is, \emph{a priori}, independent of $B$ and follows a prior  $\mathcal{P}$ satisfying the condition
\begin{equation}\label{prior-integrability}
	\exists\,\eta>0 \text{ such that}\quad \int_{\mathbb{R}}e^{\eta z^2}\mathcal{P}(dz)<\infty.
\end{equation}
Condition \eqref{prior-integrability} requires $Z$ to be sub-Gaussian, a property that will be useful for the filtering arguments later.

Consider a self-financing { portfolio strategy $\pi=(\pi_t)_{t\in[0,T]}$ that invests $\pi_t \in \mathbb{R}$ dollars} in the stock at time $t$. Then, the investor's wealth process $W^{w,\pi,Z}=(W^{w,\pi,Z}_t)_{t\in[0,T]}$ satisfies
\begin{equation}\label{wealth}
    dW^{w,\pi,Z}_t
= \bigl[r W^{w,\pi,Z}_t + (Z-r)\pi_t\bigr]\,dt
  + \sigma \pi_t\, dB_t, \quad
W^{w,\pi,Z}_0
= w > 0 .
\end{equation}
{The drift term combines the accrual at the risk-free rate $r$ on the current wealth and the excess return $(Z-r)$ on the risky investment.} In the following, we shall simply write $W^{\pi}$ to denote $W^{w,\pi,Z}$, where needed.

Because the drift $Z$ is unobservable, trading decisions must be based solely on observed stock prices $S$. Hence the investor's available information flow is the natural filtration $\mathbb{F}^S=(\mathcal{F}^S_t)_{t\in[0,T]}$ generated by $S$. The set of admissible trading strategies is
\begin{align*}
	\mathcal{A}(w):=\Big\{\pi\ \text{is $\mathbb{F}^S$-adapted},  
	{ W^{w,\pi,Z}_t}\geq  0\ \forall t\in[0,T],\ \int_0^T\pi_s^{2} ds<\infty\Big\}.
\end{align*}
Non-negativity of the wealth trajectory prevents bankruptcy and ensures that the utility functional \(U(W_T^{\pi})\) is evaluated inside the domain \(\mathbb R_+\) whenever the utility function is finite on the non-negative half-line.

We now show that uncertainty about $Z$ induces a set of beliefs about the state process $W^\pi$, denoted by $\mathfrak P$, thereby introducing ambiguity. In the ambiguity jargon  $\mathfrak{P}$ is the set of models. Throughout, we adopt the common practice of distinguishing terminologically between probability models of beliefs about the state process, which are referred to as models, and probability models that account for uncertainty about $Z$, which are referred to as priors.

For every \(Z \in\mathcal{S}\) and initial wealth $w>0$, the stochastic differential equation \eqref{wealth} admits a unique strong solution $W^{w,\pi,Z}$. By viewing $W^{w,\pi,Z}$ as a map from $\Omega$ to itself, we define the probability measure $\mathbb P^Z $ by $\mathbb{P}^{Z}\;:=\;\mathbb{P}\circ\bigl(W^{w,\pi,Z}\bigr)^{-1}$. Given $Z \in \mathcal S$, we introduce 
\begin{equation*}
    \mathfrak P:= \{\mathbb P^Z: Z\in \mathcal S\}. 
    \end{equation*}
Each realization of $Z \in \mathcal S$ determines a unique probability law of the process $W^{\pi}$, which is seen by the investor as a plausible description of how the process $W^\pi$ will evolve. If $Z$ is certain, then the set $\mathfrak{P}$ is a singleton and there is no ambiguity, i.e., no uncertainty about probabilities. In contrast, if $Z$ is a non-constant random variable, we have ambiguity.
In the remainder, we state the standard assumption that  
\begin{assumption}\label{ass:ac}
All probability measures \(\mathbb{P}^{Z}\), \(Z\in\mathcal{S}\), are absolutely continuous with respect to a reference measure \(\mathbb{P}_0\).
\end{assumption}
In particular, if $\mathcal P$ is a two-point prior, then Assumption \ref{ass:ac} is trivially satisfied by Girsanov's theorem.

Following \cite{klibanoff2005smooth}, an ambiguity-averse investor chooses $\pi\in\mathcal{A}(w)$ so as to maximise the certainty equivalent, i.e.,
\begin{equation}\label{value}
\sup_{\pi\in\mathcal{A}(w)}
\phi^{-1} \Bigl(\int_{\mathcal{S}}\phi\bigl(\mathbb{E}^{\mathbb{P}^Z}[U(W_T^\pi)]\bigr) d\mathcal{P}(z)\Bigr),
\end{equation}
where $\mathbb{E}^{\mathbb{P}^Z}$ is the expectation operator with respect to $\mathbb{P}^Z $. Here, $\phi:\mathbb{R}\to\mathbb{R}$ is non-decreasing, and either strictly concave (ambiguity-averse) or linear (ambiguity-neutral), and $U$ is the utility function of terminal wealth.

\begin{definition}[Utility function]\label{def:utility}
A \emph{utility function} is a mapping $U:(0,\infty)\to\mathbb{R}$ satisfying
\begin{enumerate}[label=(\roman*)]	\item $U$ is non-constant, increasing, and upper semicontinuous;
	\item The growth condition
	\begin{equation}\label{growth}
		\lim_{x\to\infty}\frac{U(x)}{x}=0;
	\end{equation}
	\item {$U(0):=\lim_{x\downarrow 0}U(x)$} and the finite limiting utility at infinity {$U(\infty):=\lim_{x\uparrow\infty}U(x)>-\infty$}. 
\end{enumerate}
\end{definition}
\medskip
Here, we do not assume $U$ to be concave or continuous. If $U$ is concave, \eqref{growth} is equivalent to the Inada condition $\lim_{x \to \infty}U'(x)=0$.
{Condition \eqref{growth} rules out linear utility and ensures that expected utility is finite even under unbounded upside potential of the wealth process.}  Moreover, condition (\ref{growth}) and the assumption $U(\infty)>-\infty$ imply that there is always a concave function $U_c: \mathbb{R}_+ \to \mathbb{R} $ satisfying $U_c\geq U$ (see \cite{reichlin2013utility}). 

Works on continuous-time smooth ambiguity preferences with learning, such as \cite{bauerle2024optimal} and \cite{guan2025equilibrium}, typically assume a concave utility function. Our goal is more general: we allow for non-concave preferences, including those induced by payoff structures that are nonlinear in terminal wealth, such as the option-based compensation schemes originally studied by \cite{carpenter2000does}. These contracts generate effective non-concavities even when the underlying utility function is concave (e.g., power utility). To reach this goal, we shall rely on the concavification principle.

\begin{definition}[Concave envelope]\label{def:concave-envelope}
The \emph{concave envelope} $U_c$ of $U$ is the smallest concave function {$U_c:\mathbb{R}_+\to\mathbb{R}$ such that $U_c(x)\ge U(x)$ for all $x\in\mathbb{R}_+$}.
\end{definition}
In Section \ref{sec4}, we show that the non-concave utility $U$ can be replaced by its concave envelope $U_c$, thereby transforming the nonconcave utility maximization problem to a concave one. Before such a step, we need to reformulate the ambiguity-averse utility maximization problem so that it becomes time consistent.

\section{Making Smooth Ambiguity Choices Time Consistent}\label{sec_SAP}

To obtain time-consistent preferences with ambiguity, we start from the nested functional
\[
\Gamma(U,W_T^\pi,\phi)=\phi^{-1}\!\Bigl( \int_\mathcal{S} \,
\phi\bigl( \mathbb{E}^{\mathbb{P}^Z}[\,U(W_T^\pi)\,] \bigr) d \mathcal{P}(z) \Bigr)
\]
that appears in \eqref{value}. The functional $\Gamma(U,W_T^\pi,\phi)$ is, in fact, a special case of the class of \textit{uncertainty averse preferences} proposed by \cite{cerreia2011uncertainty}. A relevant result in that work is that they obtain a preference representation in terms of quasiconcave utility functionals by relaxing the independence axiom, i.e., the coordinate independence axiom within the Savage setting. They also discuss that this is equivalent to assuming independence only at the level of risk. The key technical aspect relevant to our work is that such a preference representation enables one to express the composition ``$\phi^{-1} (\int \phi(\cdot) d\mathcal{P})$'' as a minimization problem over a set of alternative equivalent probability measures (priors) on the drift parameter space $\mathcal{S}$. This reformulation then enables the transformation of the ambiguity-averse problem into a max-min formulation, as we show next.

Let \(\mathcal{M}(\mathcal{S})\) denote the space of Borel probability measures on the compact metric space \(\mathcal{S}\), equipped with the weak topology. Intuitively, an element $\mathcal{Q}\in\mathcal{M}(\mathcal{S})$ represents a prior, that is, a belief about the parameter $Z$. As in \cite{klibanoff2005smooth}, $\mathcal{Q}$ may be viewed as a ``second order probability" over the first order probabilities (models) $\mathbb{P}^Z$. The prior $\mathcal{P}$ belongs to $\mathcal{M}(\mathcal S).$ For each $Z$, the market (state process $W^\pi$) is governed by the model \(\mathbb{P}^{Z}\) introduced in Section~\ref{sec_problem}. For any $\mathcal{Q}\in\mathcal{M}(\mathcal{S})$, we denote by $\mathbb{P}^{\mathcal{Q}}$ the probability measure on $\mathcal{B}(\mathcal{S})\times\mathcal{F}$ defined by
\begin{equation*}
\mathbb{P}^{\mathcal{Q}}(A)
\;:=\;
\int_{\mathcal{S}}\!\int_{\Omega}
\mathbf{1}_{A}(z,\omega)\,
d\mathbb{P}^{Z}(\omega)\,
d\mathcal{Q}(z),
\end{equation*}
for any  $A\in\mathcal{B}(\mathcal{S})\otimes\mathcal{F}$, and by \(\mathbb{E}^{\mathbb{P}^{\mathcal{Q}}}[\cdot]\) the corresponding expectation operator.

Combining Theorem 7 of \cite{cerreia2011uncertainty} and Theorem 6 of \cite{drapeau2013risk}, we can state the following proposition.
\begin{proposition}[Robust representation form of smooth ambiguity preferences]\label{prop3.1}
Let $\phi:\mathbb{R}\to\mathbb{R}$ be a proper, concave, non-decreasing, and upper semi-continuous function. Then the optimisation problem \eqref{value} is equivalent to
\begin{equation}\label{eq3-1}
\sup_{\pi\in\mathcal{A}(w)}\;
\inf_{\mathcal{Q}\in\mathcal{M}(\mathcal{S})}
R\bigl(\mathcal{Q},\,\mathbb{E}^{\mathbb{P}^\mathcal{Q}}[U(W_T^\pi)]\bigr),
\end{equation}
where $R:\mathcal{M}(\mathcal{S})\times\mathbb{R}\to\mathbb{R}$ is uniquely determined by $\phi$ and satisfies:
\begin{enumerate}[label=(\roman*)]\item $R(\mathcal{Q},\cdot)$ is non-decreasing and right-continuous for each fixed $\mathcal{Q}$;
\item $R$ is jointly quasi-convex in $(\mathcal{Q},s)$;
\item $\displaystyle\lim_{s\to\infty}R(\mathcal{Q}^1,s)=
\lim_{s\to\infty}R(\mathcal{Q}^2,s)$ for all
$\mathcal{Q}^1,\mathcal{Q}^2\in\mathcal{M}(\mathcal{S})$;
\item The upper envelope
$R^{+}(\mathcal{Q},s):=\sup_{s'<s}R(\mathcal{Q},s')$
is lower semi-continuous in~$\mathcal{Q}$.
\end{enumerate}

\end{proposition}
The robust representation in \eqref{eq3-1} rewrites the original nested expectation as a max-min optimization problem, where the inner infimum captures the robust adjustment of beliefs over $Z$, and the outer supremum reflects the optimal investment strategy under the resulting worst-case prior. Note that this type of robust representation is also used by \citet{mazzon2024optimal} in the context of optimal stopping.

We next show that the ambiguity-averse max-min problem in (\ref{eq3-1}) can be equivalently represented as a min-max problem. This equivalence follows from the existence of a saddle point and implies that the decision maker behaves as if she were ambiguity-neutral under an endogenously distorted belief. To allow this interchange, we need some further regularity conditions on the penalty function $R$.
\begin{assumption}\label{assume}
One of the following conditions holds:
\begin{enumerate}[label=(\roman*)]
\item $R(\mathcal{Q},\cdot)$ is continuous on $\mathbb{R}$ for every $\mathcal{Q}\in\mathcal{M}(\mathcal{S})$;
\item $R(\mathcal{Q},\cdot)$ is continuous on $(0,\infty)$ for every $\mathcal{Q}\in\mathcal{M}(\mathcal{S})$, and $\mathbb{E}^{\mathbb{P}^{\mathcal Q}}[U(W_T^\pi)]>0$ for all $Z\in\mathcal{S}$.
\end{enumerate}
\end{assumption}
These conditions hold in most practical settings and can be readily verified in specific examples (see below). We then prove that exchanging the supremum and the infimum in (\ref{eq3-1}) solves the same decision problem --- see Appendix \ref{proofthe3.1}.---
\begin{theorem}\label{theorem3.1}
Let $\phi$ satisfy the assumptions of Proposition \ref{prop3.1} and Assumption \ref{assume} hold. Then
\begin{equation}\label{eq-7-0}
    \begin{aligned}
        \sup_{\pi \in \mathcal{A}(w)} \inf_{\mathcal{Q} \in \mathcal{M}(\mathcal{S}) } R(\mathcal{Q}, \mathbb{E}^{\mathbb{P}^\mathcal{Q}}[U(W_T^\pi)] )= \inf_{\mathcal{Q}\in \mathcal{M}(\mathcal{S}) } \sup_{\pi \in \mathcal{A}(w)} R(\mathcal{Q}, \mathbb{E}^{\mathbb{P}^\mathcal{Q}}[U(W_T^\pi)] ).	
        \end{aligned}
\end{equation}

\end{theorem}
This result allows us to solve the Bayesian portfolio problem conditional on a given prior $\mathcal{Q}$ first.  In a next step, we identify the optimal prior $\mathcal{Q}^*$ arising from the outer optimization problem. The optimal strategy $\pi^*$ is then characterized as the solution to the Bayesian problem evaluated under $\mathcal{Q}^*$.

The functional in \eqref{eq-7-0} accommodates in principle any choice of ambiguity aggregator $\phi$ and penalty {$R$}, allowing great flexibility in the problem formulation.
To illustrate, we consider the families of aggregators which are the most common in the ambiguity-literature (see, e.g., \cite{taboga2005portfolio, gollier2011portfolio, ju2012ambiguity}) and for each case, we indicate which part of Assumption~\ref{assume} is satisfied, thereby showing that the regularity requirements are, in practice, harmless. It is also worth noting that $\mathcal{P}$ and $\mathcal{Q}$ in the following example must be equivalent probability measures; otherwise, the Radon-Nikodym derivative is not well defined.
\begin{enumerate}
\item \textbf{Power aggregator} \((\lambda<1,\ \lambda\neq 0)\): \(\displaystyle \phi(x)=x^{\lambda}/\lambda\) for \(x>0\).  
The associated penalty is
\begin{equation}\label{eq:power}
    R(\mathcal{Q},s)=
(s)^+\,
\Bigl(
\mathbb{E}^{\mathcal{P}}\!\Bigl[
\bigl(\tfrac{d\mathcal{Q}}{d\mathcal{P}}\bigr)^{\frac{\lambda}{\lambda-1}}
\Bigr]
\Bigr)^{\!\frac{1-\lambda}{\lambda}},
\end{equation}
where $\mathbb{E}^{\mathcal{P}}$ denotes the expectation with respect to $\mathcal P$ and $(\cdot)^+=\max\{\cdot,0\}$.
Here, the exponent $1-\lambda$ is the relative ambiguity aversion (RAA) parameter, which measures the degree of ambiguity aversion; the limit \(\lambda\to-\infty\) yields max-min preferences, while \(\lambda\to 1\) corresponds to ambiguity neutrality.  
{For any fixed \(\mathcal{Q}\) the map \(s\mapsto R(\mathcal{Q},s)\) is affine on \((0,\infty)\) and thus continuous there.  
	Consequently, Assumption~\ref{assume}\,(ii) is satisfied whenever \(\mathbb{E}^{\mathbb{P}^{\mathcal Q}}[U(W_T^\pi)]>0\).} The optimal linear-payoff problem in \cite{bauerle2024optimal} involves a power aggregator. Their expression is derived using the dual representation of the $L^p$ norm, while our formulation relies on a general robust representation (see (\ref{eq3-1})).
\item \textbf{Logarithmic aggregator:}  
$\phi(x)=\log x$ for $x>0$. The penalty takes the exponential of relative entropy:
\[
R(\mathcal{Q},s)=
(s)^+\,
\exp\!\left(-\mathbb{E}^{\mathcal{P}}\left[
\log\!\frac{d\mathcal{Q}}{d\mathcal{P}}
\right]\right).
\]
{This form connects ambiguity aversion to the Kullback-Leibler divergence, penalising beliefs that are far from the prior.} {Because the logarithm is continuous and the inside expectation is finite for every \(s\in\mathbb{R}\), the map \(s\mapsto R(\mathcal{Q},s)\) is continuous on the entire real line.  
	Hence Assumption~\ref{assume}\,(i) holds automatically.}
\item \textbf{Exponential aggregator:}  
$\phi(x) = -e^{-\gamma x}$ with $\gamma>0$. Then the penalty function is 
\begin{equation}\label{eq:exp}
    R(\mathcal{Q},s) =
s + \frac{1}{\gamma}\,
\mathbb{E}^{\mathcal{Q}}\left[
\log\!\frac{d\mathcal{Q}}{d\mathcal{P}}
\right],
\end{equation}
where $\mathbb{E} ^\mathcal{Q}$ denotes expectation with respect to $\mathcal Q$. The higher the $\gamma$ (absolute ambiguity aversion parameter (AAA)), the more ambiguity-averse the agent is. This case yields an entropic penalty, often encountered in robust control and exponential utility frameworks. {Because the logarithm is continuous and the inside expectation is finite for every \(s\in\mathbb{R}\), the map \(s\mapsto R(\mathcal{Q},s)\) is continuous on the entire real line.  Hence Assumption~\ref{assume}\,(i) holds automatically.} Indeed, this is the standard representation for the entropic risk measure (see, e.g., \cite{follmer2011entropic}). Moreover, 
\[
	\inf_{\mathcal{Q}\in\mathcal{M}(\mathcal{S})}\bigg(\mathbb{E}^{\mathbb{P}^\mathcal{Q}}[U(W_T^\pi)]+\frac{1}{\gamma}\,
	\mathbb{E}^{\mathcal{Q}}\left[
	\log\!\frac{d\mathcal{Q}}{d\mathcal{P}}
	\right] \bigg)
\]
is actually of variational form introduced by \cite{maccheroni2006ambiguity}.
	\end{enumerate}
We note that the penalty functions induced by all standard smooth-ambiguity aggregators satisfy the continuity requirements in Assumption~\ref{assume}, and in particular the power, logarithmic, and exponential aggregators illustrated above, satisfy the assumptions of Theorem \ref{theorem3.1}. The minimax interchange is thus justified for the specifications most frequently used in applications. It is worth noting that, under commonly used choices of the aggregator $\phi$, the functional $R$ is linear in $\mathbb{E}^{\mathbb{P}^{\mathcal{Q}}}[\cdot]$. As a result, the problem reduces to a time-consistent Bayesian adaptive control problem, which can be solved using nonlinear filtering and martingale methods.

The results have a notable implication: the inner problem in \eqref{eq-7-0} is of the form 
\begin{equation*}
	\sup_{\pi \in \mathcal{A}(w)} R(\mathcal{Q}, \mathbb{E}^{\mathbb{P}^\mathcal{Q}}[U(W_T^\pi)] ).
\end{equation*}
which, under linearity of $R(Q,s)$ corresponds to that of an ambiguity-neutral decision-maker. Specifically, it becomes
\begin{equation}\label{eq:ambneutr}
    \sup_{\pi \in \mathcal{A}(w)} \mathbb{E}^{\mathbb{P}^\mathcal{Q}}[U(W_T^\pi)].
\end{equation}
This is the case of the three aggregator functions illustrated above. The approach then transforms the problem into two main steps. In the first, we solve a classical Bayesian decision problem in which ambiguity does not play any role and prior $\mathcal{Q}$ is updated through observations of market prices. This step has the theoretical implication of making the problem dynamically consistent. It also has the practical advantage that ambiguity-neutral problems are, in general, easier to solve than the corresponding ambiguous counterparts. Ambiguity preferences enter only through the outer minimization over priors. Thus, one is free to select alternative forms of the ambiguity functional and apply them to the solution. Overall, this points to greater tractability and modeling flexibility.

In the next section, we apply the approach to the study of the class of decision delegated portfolio management problems. Our goal is not only to characterize the optimal investment strategies, but also to identify the endogenous prior induced by the investor's ambiguity preferences.
 
\section{Application:\ Convex Incentive Contracts with Smooth Ambiguity Preferences}\label{sec4}

In Section \ref{sec:steps}, we present the steps implied by Proposition \ref{prop3.1} and Theorem \ref{theorem3.1}, which have a general validity beyond delegated portfolio management. In Section \ref{sec:4-2}, we specify the steps to delegated portfolio management.

\subsection{The Theoretical Steps} \label{sec:steps}

Given Problem (\ref{value}), we outline the analysis procedure in the following steps: 
\begin{enumerate}[label=\textbf{Step \arabic*:}, leftmargin=*]
   \item Apply Proposition~\ref{prop3.1}, to obtain a robust representation of the ambiguity-averse problem in the form of a max--min optimization.

\item Apply Theorem~\ref{theorem3.1}, to turn the max--min problem into a min--max problem. 
\item Choose a specification of the aggregator $\phi$, such that the functional $R$ is linear in its second argument, namely in $\mathbb{E}^{\mathbb{P}^{\mathcal{Q}}}[\cdot]$.

\item Conditional on a given prior $\mathcal{Q}$ solve the resulting ambiguity-neutral problem \eqref{eq:ambneutr} using nonlinear filtering and martingale methods, and derive the optimal trading strategy $\pi^*(\mathcal{Q})$.

\item We then solve the outer optimization problem over the set of priors to identify the worst-case prior $\mathcal{Q}^*$. The optimal trading strategy is obtained by evaluating $\pi^*(\mathcal{Q})$ at $\mathcal{Q}=\mathcal{Q}^*$.
\end{enumerate}

\subsection{Modeling Delegated Portfolio Management with Nonlinear Payoffs}\label{sec:4-2}

Delegated fund managers are often compensated through convex schemes, most prominently option-like contracts. In such contracts the payment consists of a fixed base fee plus a performance-based component that resembles a call option written on the managed fund and struck at a benchmark value at maturity (see, e.g., \cite{carpenter2000does}). We study the optimal investment policy of a manager who faces ambiguity about the return of a risky asset and is compensated through an option-based payoff, within the framework of smooth ambiguity preferences.
 This extends \cite{carpenter2000does} by allowing for ambiguity and lets us quantify how ambiguity aversion shapes portfolio choices.

 Specifically, the manager's payoff at time $T$ is $\delta \in (0,1]$ shares of a call option on the fund with strike $K$, plus a constant base $C \geq 0$. The payoff function $g$ is
\begin{equation}\label{eq:payoff}
    g(W_T^{\pi}):=\delta (W_T^{\pi}-K)^+ +C.
\end{equation}
The manager is ambiguity-averse with respect to the uncertain return of the risky asset and evaluates terminal wealth utility using a smooth ambiguity functional. Moreover, she updates her prior based on the information of stock prices.  Her decision problem is to determine
\begin{equation}\label{eq4-1}
	\sup_{\pi\in\mathcal A(w)}
	\phi^{-1}
	\left(
	\int_{\mathcal S}
	\phi \Bigl(\mathbb{E}^{\mathbb{P}^Z}\bigl[u\bigl(g(W_T^\pi)\bigr)\bigr]\Bigr)
	d\mathcal P(z)
	\right),
\end{equation}
which is a special case of \eqref{value} with $U(\cdot)=u\!\circ\! g(\cdot)$. 

\begin{figure}[htbp]
	\begin{minipage}[t]{0.47\textwidth}
		\centering
		\includegraphics[width=\textwidth]{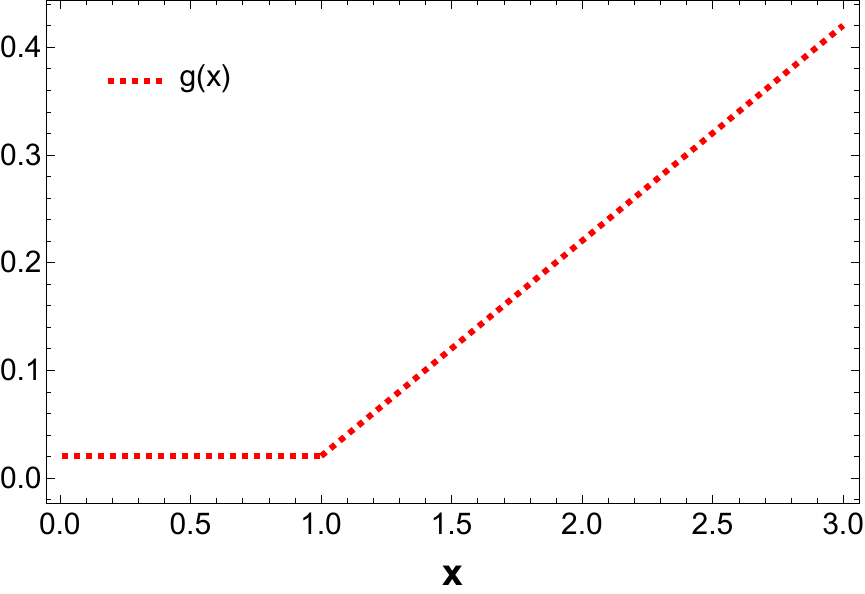}
		\caption{Manager's payoff function with $\alpha=0.5, \delta=0.2, K=1, C=0.02$.}
		\label{p1}
	\end{minipage}
	\hfill
	\begin{minipage}[t]{0.47\textwidth}
		\centering
		\includegraphics[width=\textwidth]{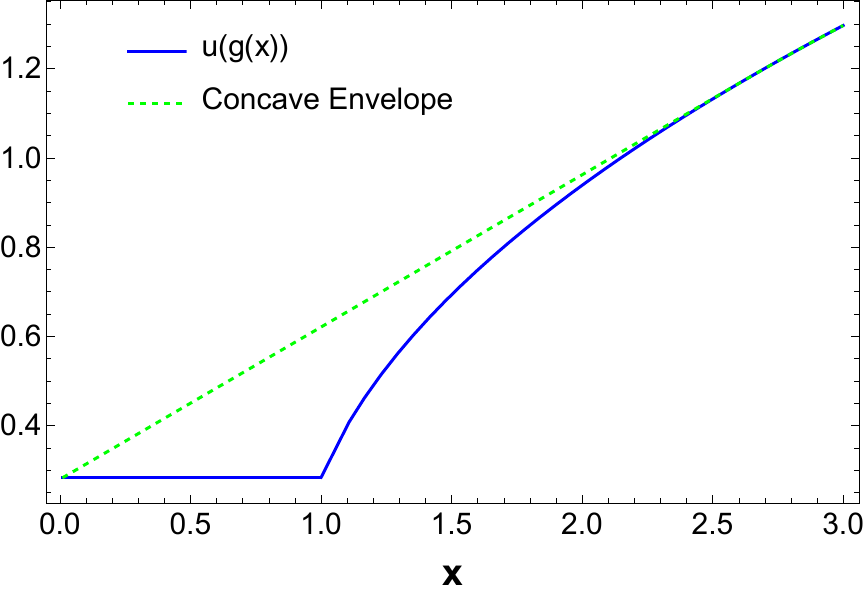}
		\caption{Manager's original and concavified objective functions with $\alpha=0.5, \delta=0.2, K=1, C=0.02$.}
		\label{p2}
	\end{minipage}
\end{figure}
\begin{example}
    Consider a power utility function of the form 
\begin{equation*}
    u(x)=x^\alpha/\alpha,
\end{equation*}
where $x\in[0,\infty)$, with $\alpha<1$, $\alpha\neq0$, so that
\begin{equation*}
   \mathrm{RRA}:=1-\alpha,
\end{equation*}
is the coefficient of relative risk aversion, together with the payoff function $g(\cdot)$ given in (\ref{eq4-1}). Figures~\ref{p1} and~\ref{p2} display the manager's payoff function $g(\cdot)$ (right panel, dotted red line), the induced utility $u\circ g(\cdot)$ (left panel, solid blue line), and its concave envelope (left panel, dashed green line) for selected parameter values. The figures illustrate that the payoff function is nonlinear and that the induced utility $U(\cdot)=u \circ g(\cdot)$ is not globally concave. Hence, the specification satisfies Definition~\ref{def:utility}.

\end{example}

We now adopt the steps highlighted in Section \ref{sec:steps} to solve the non-concave portfolio management problem in \eqref{eq4-1}. In Step 1, we apply Proposition \ref{prop3.1} to obtain the robust representation:
\begin{equation*}
	\sup_{\pi\in\mathcal{A}(w)}\;
	\inf_{\mathcal{Q}\in\mathcal{M}(\mathcal{S})}
	R\bigl(\mathcal{Q},\,\mathbb{E}^{\mathbb{P}^\mathcal{Q}}[u(g(W_T^\pi))]\bigr).
\end{equation*}
In Step 2, by Theorem \ref{theorem3.1} we obtain the following equivalent min-max optimization problem:
\begin{equation}\label{eq-10-1}
    \begin{aligned}
\inf_{\mathcal{Q}\in \mathcal{M}(\mathcal{S}) } \sup_{\pi \in \mathcal{A}(w)} R(\mathcal{Q}, \mathbb{E}^{\mathbb{P}^\mathcal{Q}}[u(g(W_T^\pi))] ).    
    \end{aligned}
\end{equation}
In Step 3, regarding the choice of the function $\phi$, we focus on two economically relevant cases, the power--power case in which both \(\phi\) and \(u\) are CRRA, and the exponential--power in which \(\phi\) is exponential and \(u\) is CRRA. These two specifications are widely used in the literature, including one-period asset allocation problems (e.g., \cite{taboga2005portfolio,gollier2011portfolio}) and consumption-based asset pricing models (e.g., \cite{ju2012ambiguity,collard2018ambiguity}). 

Formally, in the power-power case, we write $\phi(x)=\frac{x^\lambda}{\lambda}$, $\lambda<1$, $\lambda\neq0$, and $u(x)=\frac{x^\alpha}{\alpha}$, $\alpha<1$, $\alpha\neq0$. Then, by Equation \eqref{eq:power}, the decision-problem (\ref{eq-10-1}) becomes
	\begin{equation}\label{eq4-4}
	\inf_{\mathcal{Q}\in \mathcal{M}(\mathcal{S}) }\sup_{\pi\in\mathcal{A}(w)}\;
		\bigg(\mathbb{E}^{\mathbb{P}^\mathcal{Q}}[u(g(W_T^\pi))]\bigg)^+ \left(\mathbb{E}^{\mathcal{P}}\left[
		\left(\tfrac{d\mathcal{Q}}{d\mathcal{P}}\right)^{\frac{\lambda}{\lambda-1}}
		\right]\right)^{\!\frac{1-\lambda}{\lambda}}.		
	\end{equation}
    In the exponential-power case, we write $\phi(x)=-e^{-\gamma x}$, $\gamma>0$, and $u(x)=\frac{x^\alpha}{\alpha}$, $\alpha<1$, $\alpha\neq0$. Then, by Equation \eqref{eq:exp}, the decision-problem (\ref{eq-10-1}) becomes
	\begin{equation}\label{eq4-5}
	\inf_{\mathcal{Q}\in \mathcal{M}(\mathcal{S}) } \sup_{\pi\in\mathcal{A}(w)}\;
		\bigg\{ \mathbb{E}^{\mathbb{P}^\mathcal{Q}}[u(g(W_T^\pi))] + \frac{1}{\gamma}\,
		\mathbb{E}^{\mathcal{Q}}\left[
		\log\!\frac{d\mathcal{Q}}{d\mathcal{P}}
		\right] \bigg\}.
	\end{equation}	
Both \eqref{eq4-4} and \eqref{eq4-5} are linear in $\mathbb{E}^{\mathbb{P}^\mathcal{Q}}[u(g(W_T^\pi))]$. This allows us to solve the ambiguity neutral problem, which now takes the form
\begin{equation}\label{eq4-7}
	V(w;\mathcal Q):=	\sup_{\pi \in \mathcal{A}(w)} \mathbb{E}^{\mathbb{P}^\mathcal{Q}}[u(g(W_T^\pi))].
\end{equation}
Technically, this is a Bayesian adaptive problem with a nonlinear payoff $g$. In this respect, our approach parallels that of \cite{bauerle2024optimal} for the linear payoff case, extending it to option-type payoffs. We then use stochastic filtering theory (e.g., \cite{karatzas2001bayesian, rieder2005portfolio}) to reduce the partial-information problem with an unknown expected return to one with an adapted (observable) drift. This is achieved by updating priors based on the observed price process (see Appendix~\ref{filtering} for details). We subsequently solve the resulting complete-information problem using martingale methods (e.g., \cite[Ch.~3]{karatzas1998methods}) combined with the concavification principle.

Now we come to Step 4, related to the financial modeling part of the problem. 
Define the market price of risk by \(\Theta = (Z-r)/\sigma\), and denote its realization by \(\theta = (z-r)/\sigma\). 
Let \(I\) be the inverse of the marginal utility function \(u'\), given by \(I(x) := x^{\frac{1}{\alpha-1}}\), and
\[
\xi_T := \frac{e^{-rT}}{F(T,Y_T;\mathcal Q)}
\]
denote the state-price density, where \(Y_t\) is a \(\mathbb P\)-Brownian motion with drift \(\Theta\), and $F(t,y;\mathcal Q):= \int_{\mathcal S} \exp\!\left\{\theta y - \tfrac{1}{2}\theta^2 t \right\}
\,\mathcal Q(d\theta).$ The state-price density $\xi_T(\omega)$ represents the (discounted) price at $t=0$ of a unit payoff received in a given state $\omega$ at time $T$, reflecting both time discounting and the belief-adjusted pricing of risk. Here, a state refers to a realization of the underlying market uncertainty at time $T$. A large $\xi_T$ means that a unit payoff delivered in that state is expensive today: investors must pay more now to receive one unit of payoff there, which is why such states are interpreted as unfavorable states (see \cite{carpenter2000does} for a detailed discussion).

We then prove the following result regarding the optimal terminal wealth (denoted by $W^*_T$) and optimal investment strategy (denoted by $\pi_t^\ast$) for an ambiguity neutral decision-maker (see Appendix \ref{proofthe4.1} for the proof).

\begin{theorem}\label{the4.1}
	 With the above definitions, we have the following:
     \begin{enumerate}[label=(\roman*)]
     \item The optimal terminal wealth for the ambiguity-neutral decision-problem in \eqref{eq4-7} is
	\begin{equation}\label{eq4-3}
		W^*_T(\kappa^*_w, {\xi}_T)=\mathcal X(\kappa^*_w \xi_T)=
        \begin{cases}
           h(\kappa^*_w {\xi}_T), \quad  0<{\xi}_T<\frac{\widehat{y}}{\kappa^*_w},\\
           0, \quad \xi_T \geq \frac{\widehat{y}}{\kappa^*_w}, \end{cases}		
	\end{equation}
	where the function $h(x):=\dfrac{I(x/\delta)-C}{\delta}+K $ with $x\in[0,\infty)$, and we recall that $K$, $\delta$ and $C$ are defined in \eqref{eq:payoff}, and where the concavification point $\widehat{y}\in (0, \delta C^{\alpha-1})$ is the unique root of 
	\[
		u(g(h(y)))-u(C) =y h(y). 
\]
\item For any initial wealth $w>0$, $\kappa^*$ is the unique root of the budget constraint 
	\[
		\mathbb{E}^{\mathbb{P}}[\xi_T W^*_T(\kappa^*,\xi_T)]=w,
	\]
	where $\mathbb E^{\mathbb P}$ is the expectation with respect to $\mathbb P$ and we write $\kappa_w^* := \kappa^*(w)$ to emphasize its dependence on the initial wealth $w$.
    
    \item The optimal fraction invested in the stock at time $0<t\leq T$ for (\ref{eq4-7}) is 
{ \begin{equation}\label{eq4-9}
\frac{\pi_t^*}{W^*_t}
=
-\frac{\kappa^*_w e^{-rT}}{\sigma}
\frac{
\displaystyle
\int_{\mathbb{R}}
(\mathcal X)'\!\left(
\frac{\kappa^*_w e^{-rT}}{F(T,Y_t+z;\mathcal Q)}
\right)
\frac{\nabla F(T,Y_t+z;\mathcal Q)}
     {F(T,Y_t+z;\mathcal Q)^2}\varphi_{T-t}(z)dz
}{
\displaystyle
\int_{\mathbb{R}}
\mathcal X\!\left(
\frac{\kappa^*_w e^{-rT}}{F(T,Y_t+z;\mathcal Q)}
\right)
\varphi_{T-t}(z)\,dz
}.
\end{equation}
}

     \item The ambiguity-neutral value function at the optimum is
{	\begin{equation}\label{eqvalue}
V(w;\mathcal Q)
=
\int_{\mathbb{R}}
F(T,z;\mathcal Q)\,
u\!\left(
g\!\left(
\mathcal X\!\left(
\frac{e^{-rT}\kappa^*_w}{F(T,z;\mathcal Q)}
\right)
\right)
\right)
\varphi_T(z)\,dz .
\end{equation}
}
	where $\varphi_T$ is the density of $\mathcal{N}(0,T)$.
     \end{enumerate}
\end{theorem}
A key qualitative implication of \eqref{eq4-3} is that the optimal terminal wealth is \emph{discontinuous} in the state-price density \(\xi_T\) and exhibits an ``all-or-nothing'' structure: whenever \(\xi_T \ge \widehat y/\kappa^*_w\), the optimal terminal payoff is zero. Here, ``zero payoff" means that the investor ends the horizon with zero terminal wealth in those states. In favorable financial scenarios (low \(\xi_T\)), the optimal trading policy resembles the classical Merton rule (see \cite{merton1971optimum}). In sufficiently unfavorable scenarios (high \(\xi_T\)), however, the optimal terminal wealth collapses to zero. Economically, high \(\xi_T\) corresponds to states in which delivering one unit of terminal payoff is particularly costly in present-value terms. Under a call-type payoff, compensation is truncated below the performance threshold \(K\), so additional resources allocated to such states do not increase \(g(W_T^{\pi})\) and hence yield no marginal benefit. The optimal policy therefore does not hedge these expensive states and instead reallocates resources toward inexpensive, favorable states, consistent with \cite{chen2024equivalence}.

Our framework subsumes both the frameworks of \cite{bauerle2024optimal} and \cite{carpenter2000does}. Setting \(K=0\), \(C=0\), and \(\delta=1\) in \eqref{eq:payoff} we recover the linear-payoff specification of \cite{bauerle2024optimal}. On the other hand, if \(Z\) is assumed to be a known constant, the optimal terminal wealth simplifies to
\begin{equation}\label{noambiguity}
		W^*_T(\kappa^*_w, {\xi}_T)=\mathcal X(\kappa^*_w \xi_T)=
        \begin{cases}
           h(\kappa^*_w {\xi}_T), \quad  0<{\xi}_T<\frac{\widehat{y}}{\kappa^*_w},\\
           0, \quad \xi_T \geq \frac{\widehat{y}}{\kappa^*_w}, \end{cases}		
	\end{equation}
where $\xi_T = \exp\{-rT - \tfrac{1}{2}\theta^2 T - \theta W_T\}.$ This specification recovers Theorem~1 in \cite{carpenter2000does}.

Finally (Step 5), we solve the outer optimization problem in \eqref{eq-10-1} over the set of priors to identify the worst-case prior \(\mathcal{Q}^*\) under the two proposed specifications. We start with the power-power case.

\begin{proposition}[Power--power specification] \label{prop4.1}
When  $$\phi(x)=\frac{x^\lambda}{\lambda},\ \text{with}\ \lambda<1, \lambda \neq 0,   $$ and $$u(x)=\frac{x^\alpha}{\alpha},\ \text{with} \  \alpha<1,\alpha\neq 0,$$
the optimal investment strategy $\pi^*$ satisfies \eqref{eq4-9} with optimal prior $\mathcal Q^*$  solves  
    \begin{equation}\label{eqpower}
    \begin{aligned}
            \inf_{\mathcal Q  \in \mathcal M(\mathcal S)}\Big(V(w;\mathcal Q)\Big)^+ \left(\mathbb{E}^{\mathcal{P}}\left[
		\left(\tfrac{d\mathcal{Q}}{d\mathcal{P}}\right)^{\frac{\lambda}{\lambda-1}}
		\right]\right)^{\!\frac{1-\lambda}{\lambda}} = \Big(V(w;\mathcal Q^*) \Big)^+ \left(\mathbb{E}^{\mathcal{P}}\left[
		\left(\tfrac{d\mathcal{Q^*}}{d\mathcal{P}}\right)^{\frac{\lambda}{\lambda-1}}
		\right]\right)^{\!\frac{1-\lambda}{\lambda}},
        \end{aligned}
    \end{equation}
	where $V(w;\mathcal Q)$ is given by \eqref{eqvalue}.
\end{proposition}
Even though both our framework and \cite{bauerle2024optimal} employ a power aggregator \(\phi\) to model ambiguity aversion and power utility \(u\) to capture risk aversion, two key differences arise. First, we allow for a nonlinear payoff structure, which generates the endogenous ``all-or-nothing" pattern in optimal terminal wealth. Such behavior does not arise in \cite{bauerle2024optimal}, where the payoff is linear and the objective remains concave. This nonlinearity yields new economic insights, as illustrated in Section~\ref{sec:numerical}, particularly regarding how ambiguity aversion interacts with convex incentive contracts. Second, from a methodological perspective, we rely on a general robust representation rather than the dual representation of the \(L^p\) norm used in \cite{bauerle2024optimal}. This approach provides greater flexibility in the specification of ambiguity preferences and extends naturally to non-concave payoff environments.

We now consider the exponential-power specification.
\begin{proposition}[Exponential--power specification] \label{prop4.2}
	When  $$\phi(x)=-e^{-\gamma x},\ \text{with}\ \gamma>0,   $$ and $$u(x)=\frac{x^\alpha}{\alpha},\ \text{with} \  \alpha<1,\alpha\neq 0,$$ the optimal investment strategy $\pi^*$ satisfies \eqref{eq4-9} with optimal prior $\mathcal Q^*$  given by 
\begin{equation}\label{eqexp}
    \begin{aligned}
        	\inf_{\mathcal Q  \in \mathcal M(\mathcal S)}
		\bigg\{ V(w;\mathcal Q) + \frac{1}{\gamma}\,
		\mathbb{E}^{\mathcal{Q}}\left[
		\log\!\frac{d\mathcal{Q}}{d\mathcal{P}}
		\right] \bigg\}= \bigg\{ V(w;\mathcal Q^*) + \frac{1}{\gamma}\,
		\mathbb{E}^{\mathcal{Q}^*}\left[
		\log\!\frac{d\mathcal{Q}^*}{d\mathcal{P}}
		\right] \bigg\},
        \end{aligned}
\end{equation}
	where $V(w;\mathcal Q)$ is given by \eqref{eqvalue}.
\end{proposition}
Equations \eqref{eqpower} and \eqref{eqexp} imply that solving the original smooth ambiguity problem in \eqref{eq4-1} amounts to characterizing the optimal investment strategy in \eqref{eq4-9} for a given prior $\mathcal{Q}$ and identifying the corresponding optimal prior $\mathcal{Q}^*$. Although we focus on two specific ambiguity aggregators, the structure of the formulation makes the framework readily extensible. Once the ambiguity-neutral benchmark in \eqref{eq4-7} is characterized, alternative ambiguity specifications can be chosen by the analyst and incorporated without difficulty. This modular structure represents a key methodological advantage, enabling different ambiguity attitudes to be imposed in a tractable and systematic manner.

Obtaining a closed-form expression for $\mathcal{Q}^*$ is generally out of reach due to the difficulty of determining $V(w;\mathcal{Q})$ in \eqref{eqvalue}. However, the solution can be found numerically. In the next section, we use computational experiments to analyze how ambiguity aversion affects the optimal investment strategies.

\section{Numerical Experiments and Comparative Statics}\label{sec:numerical} 
In this section, we quantify how ambiguity aversion shapes the manager's priors, terminal wealth, and optimal investment strategies, and we relate the key comparative statics to the theoretical results in Sections~\ref{sec_SAP} and~\ref{sec4}. Solving the two outer minimization problems over the set of priors also allows us to compare the distorted priors implied by the different ambiguity attitudes, offering a further interpretation perspective for Propositions \ref{prop4.1} and \ref{prop4.2}. All computations in this section are performed using Mathematica 13.3. The code is available upon request.

\paragraph{Ambiguity Absence and the Effects of Nonlinear Payoffs.} We start by considering the problem in the absence of ambiguity, which we will also use as a reference benchmark. Specifically, we assume that the expected return Z is known to the manager and equal to $z_0$. 
\begin{table}[htbp]
	\centering
    \caption{Basic parameter set in the delegated portfolio management.}
	\label{tab1}	
	\begin{tabular}{ccccccccc}
		\toprule
		Parameter&$z_0$ &  $r$	& $\sigma$  & $\delta$ & $K$ & $C$ & $T$  \\ 
		\midrule
		Value & 0.078& 0.02 & 0.3 & 0.2 & 1 & 0.02 & 10 \\
		\bottomrule
	\end{tabular}
\end{table}
Table~\ref{tab1} reports the financial model parameters. We set $z_0=0.078$, the risk-free rate to $r = 0.02$, and the volatility to $\sigma = 0.3$, which together imply a reasonable Sharpe ratio of approximately $20\%$. The compensation scheme is specified as 
\[
g(W_T^\pi) = 0.2 (W_T^\pi - 1)^+ + 0.02,
\]
representing an option-based payoff. The slope $\delta = 0.2$ captures a moderate incentive intensity, while the fixed component $C = 0.02$ ensures participation and provides downside protection. The contract horizon is set to $T = 10$ years.

We first study the relationship between the optimal terminal wealth $W^*_T$ and the state-price density $\xi_T$. 
Implementing \eqref{noambiguity}, we compute the optimal terminal wealth numerically.  Figure~\ref{fig3} displays the relationship between optimal terminal wealth (vertical axis, $W_T^\ast$) and the state price density (horizontal axis, $\xi_T$) for three levels of relative risk aversion (RRA): $\mathrm{RRA}=0.5$ (dashed, red), $\mathrm{RRA}=0.6$ (dotted, green), and $\mathrm{RRA}=0.7$ (solid, blue).
\begin{figure}[htbp]
	\centering
	\begin{minipage}[t]{0.47\textwidth}
		\includegraphics[width=\textwidth]{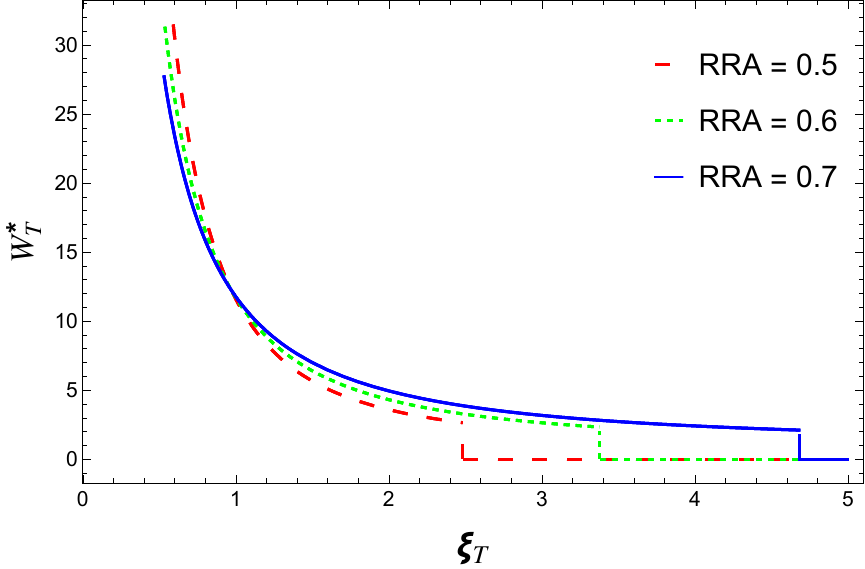}
		\caption{Optimal terminal wealth $W_T^*$ as a function of $\xi_T$ for three relative risk aversion ($\mathrm{RRA=1-\alpha}$) levels (initial wealth $w=10$).}
		\label{fig3}	
	\end{minipage}
	\hfill
	\begin{minipage}[t]{0.47\textwidth}
		\includegraphics[width=\textwidth]{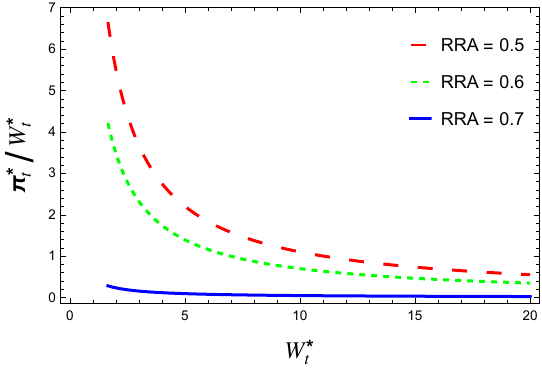}
		\caption{Optimal risky proportion vs.\ wealth $W^*_t$ for three relative risk aversion ($\mathrm{RRA=1-\alpha}$) levels (initial wealth $w=10$, $t=T-1$).}
		\label{fig4-1}
	\end{minipage}
\end{figure}

The results in Figure~\ref{fig3} show that $W^*_T$ is discontinuous in $\xi_T$ for all three levels of relative risk aversion: it decreases until $\xi_T$ reaches $\widehat y/\kappa^*_w$, after which it drops to zero. This kink reflects that, once the option is sufficiently out of the money, additional wealth barely changes the compensation, so it is optimal for the manager to cease allocating resources to these states. Intuitively, in states with high state-price density (``expensive'' states), extra wealth hardly increases compensation because the option-style payoff is locally flat. The investor therefore shifts resources away from these expensive states toward low-$\xi_T$ (``inexpensive'') states where each dollar has more impact. This ``all-or-nothing'' structure in $W^*_T$ is a direct effect of non-concavity.
Moreover, Figure~\ref{fig3} shows that greater risk aversion shifts the kink to the right. Greater risk aversion raises the effective penalty on dispersion across states. As a result, in favorable (low-$\xi_T$) states, a more risk-averse investor chooses a lower payoff than a less risk-averse one, while in unfavorable states (high-$\xi_T$) the zero-payoff region shrinks. Thus higher RRA compresses the cross-state wealth distribution, replacing ``lottery-like'' payoffs by more concave profiles, consistent with the concavification logic in Theorem~\ref{the4.1}.

Regarding the trading strategy, Figure~\ref{fig4-1} plots the optimal trading ratio (vertical axis, $\pi_t^\ast / W_t^\ast$) as a function of current wealth (horizontal axis, $W_t^\ast$) at time $t = T-1 = 9$ in the absence of ambiguity. As expected, greater risk aversion lowers the fraction of wealth invested in the risky asset. Quantitatively, a 40\% increase in RRA reduces the risky share by roughly one third at intermediate wealth levels, illustrating that the non-concave payoff preserves the standard risk-return trade-off but amplifies its impact near the extremes.

\paragraph{Ambiguity Aversion Interacts with Nonlinear Payoffs.} To study the impact of ambiguity presence, we consider that the decision-maker is no longer certain about the underlying value of $Z$ and models this uncertainty, allowing $Z$ to become a random variable over two scenarios, an \textit{unfavourable} scenario over which $Z=z_1$ and a  \textit{favourable} one with $Z=z_2$. The decision-maker assigns a prior $\mathcal P$ over these realizations, with $\mathcal P(Z=z_1)=q$, and $\mathcal P(Z=z_2)=1-q$. Numerically, we assign  $z_1=0.03$ and $z_2=0.09$, and $\mathcal P(Z=z_1)=0.2 $ and $\mathcal P(Z=z_2)=0.8$, so that $\mathbb E^{\mathcal P}[Z]=0.078=z_0$, the benchmark return $z_0$ in the ambiguity-neutral experiments. Thus, the decision-maker at time $t=0$  is optimistic about the possibility of the more favourable scenario, assigning it an 80\% probability, against a 20\% chance of the lower return scenario.

The first effect of ambiguity preferences is a distortion of the decision-maker's belief about the unknown parameter $Z$. This amounts to understanding how the distorted prior implied by smooth ambiguity preferences (Equation \eqref{eq-10-1}).

\begin{table}[htbp]
	\centering
    \caption{Optimal $q^*$ for different RAA levels in the power-power case (initial wealth $w=10$, $\mathrm{RRA = 0.5}$)}
	\label{tab3}	
    \begin{tabular}{ccccccccccc}
		\toprule
		$1-\lambda$ (RAA) & 0.01 & 0.02 & 0.04& 0.1 & 0.3 & 0.7 & 1.2 & 1.7 &2.0&2.2 \\ 
		\midrule
		$q^*$ & 0.733 & 0.573 & 0.413 &0.238 &0.168 & 0.127 &  0.102 & 0.091&0.087&0.086 \\
		\bottomrule
	\end{tabular}
\end{table}

We begin with the power-power specification in Proposition~\ref{prop4.1}. Under ambiguity aversion, the distorted belief $\mathcal Q^\ast$ assigns probability $q^\ast := \mathcal Q^\ast(Z=z_2)$ and $1-q^\ast = \mathcal Q^\ast(Z=z_1)$, where the distorted belief $\mathcal Q^\ast$ arises endogenously from the outer maximization problem in~\eqref{eqpower}.
The resulting numerical value depends on the relative ambiguity aversion $1-\lambda$. Table~\ref{tab3} reports the values of $q^\ast$ after solving Problem \eqref{eqpower} for increasing values of $1-\lambda$, i.e., for higher ambiguity aversion. To illustrate, at $1-\lambda=0.01$, we have $q^\ast=0.733$, while at $1-\lambda=2.2$ we find $q^\ast=0.086$. Because $q^\ast$ represents the probability of the favorable state, the values in Table~\ref{tab3} show that decision-makers shift their belief toward considering the bad state more likely as ambiguity aversion increases. This distortion should imply a cautious approach in the face of uncertainty.

\paragraph{Sensitivity Analysis: A $2^3$ Factorial Design.}
We quantify the sensitivity of belief distortions using a $2^3$ full-factorial design that varies (A) the baseline prior over the good state (low $q=0.5$ vs.\ high $q=0.8$), (B) relative risk aversion (RRA $=0.3$ vs.\ $0.5$), and (C) relative ambiguity aversion (RAA $=0.01$ vs.\ $0.3$).
For each of the eight configurations, we solve the Problem \eqref{eqpower} and record the optimal distorted good-state probability $q^*$. Table~\ref{table4} reports the implied decomposition and Table \ref{table5} reports the eight values in Appendix \ref{design}.

\begin{table}[htbp]
\centering
\caption{Factorial effects on the distorted good-state probability $q^*$}
\label{table4}
\begin{tabular}{lcc}
\hline
 & Effect on $q^*$ & Economic interpretation \\
\hline
Prior optimism (A) & $+0.160$ & Anchors distorted beliefs toward the good state \\
Risk aversion (B) & $+0.035$ & Small impact \\
Ambiguity aversion (C) & $-0.447$ & Pessimistic tilt: shifts mass to the bad state \\
\hline
A $\times$ B & $+0.001$ & Negligible \\
A $\times$ C & $-0.113$ & Ambiguity attenuates prior anchoring \\
B $\times$ C & $+0.021$ & Weak interaction \\
A $\times$ B $\times$ C & $+0.033$ & Small higher-order effect \\
\hline
\end{tabular}
\end{table}

Three patterns emerge. First, ambiguity aversion is the dominant force shaping distortions: moving from low to high ambiguity aversion reduces $q^*$ by about $0.45$ on average, implying a strong pessimistic tilt (probability mass shifts from the good to the bad state). Second, a more optimistic prior increases $q^*$ (about $0.16$ on average), indicating that distorted beliefs remain anchored to baseline information. Third, ambiguity aversion attenuates the role of the prior: the negative $A\times C$ interaction shows that prior optimism translates into a smaller increase in $q^*$ when ambiguity aversion is high. Risk aversion has only a minor direct or interactive impact. Overall, belief distortions are driven primarily by ambiguity attitudes, while prior information matters mainly when ambiguity concerns are weak.

\begin{figure}[htbp]
	\centering
	\begin{minipage}[t]{0.47\textwidth}
			\centering
			\includegraphics[width=\textwidth]{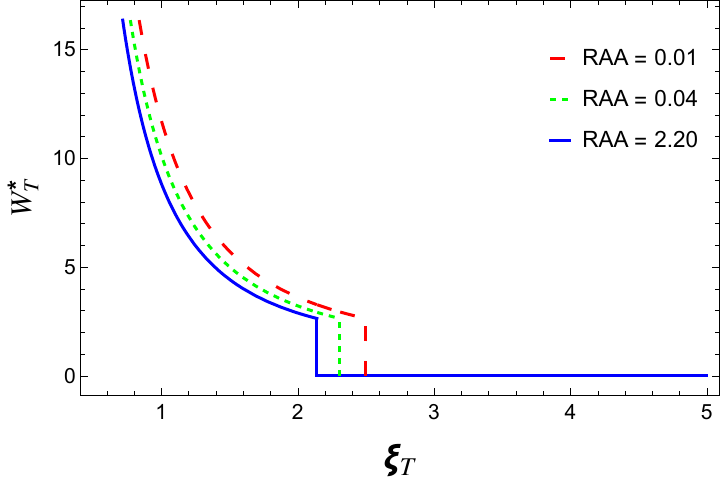}
			\caption{Optimal terminal wealth $W_T^*$ as a function of $\xi_T$ for three ambiguity aversion levels ($\mathrm{RAA=1-\lambda}$) in power-power case (initial wealth $w=10$, $\mathrm{RRA = 0.5}$).}
			\label{fig5}
		\end{minipage}		
	\hfill
	\begin{minipage}[t]{0.47\textwidth}
		\includegraphics[width=\textwidth]{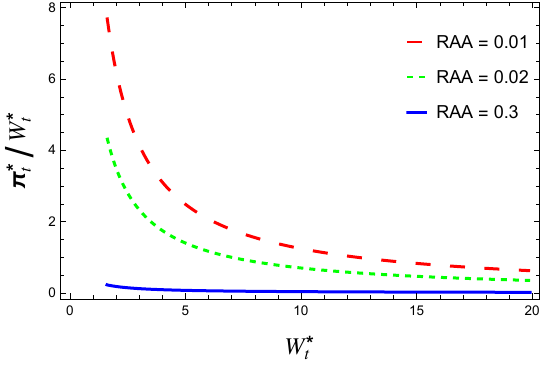}
		\caption{Optimal risky proportion vs.\ wealth $W^*_t$ for three ambiguity aversion levels ($\mathrm{RAA=1-\lambda}$) in the power-power case (initial wealth $w=10$, $\mathrm{RRA = 0.5}$, $t=T-1$).}
		\label{fig:pit}
	\end{minipage}
\end{figure}
We can now examine the impact of ambiguity preferences on the optimal terminal wealth $W_T^\ast$ and trading strategy $\pi_t^\ast$. In Figure~\ref{fig5}, we display the optimal terminal wealth (vertical axis, $W_T^\ast$) against the state-price density (horizontal axis, $\xi_T$) for three levels of ambiguity aversion: 
$\mathrm{RAA}=0.01$ (dashed, red), $\mathrm{RAA}=0.04$ (dotted, green) and $\mathrm{RAA}=2.20$ (solid, blue). Figure \ref{fig5} shows that $W_T^*$ decreases with $\xi_T$ until it reaches the cutoff $\widehat y/\kappa^*_w$, after which it drops to zero for all three values of RAA. This behavior is similar to the one in the no-ambiguity case. Ambiguity aversion has two effects. First, stronger ambiguity aversion shifts the cutoff to the left, and, second, it makes the curve $W_T^*$ systematically lower than the corresponding curve with weaker ambiguity aversion. Overall, the region in which the option is ``effectively in the money'' shrinks as beliefs become more pessimistic, and the manager accepts zero payoff already at intermediate $\xi_T$. Intuitively, a more ambiguity-averse manager assigns a higher weight to adverse drifts, which raises the shadow price $\kappa^*_w$, so the interior solution $W_T^*=\mathcal X(\kappa^*_w\xi_T)$ falls and the zero-payoff region expands. This mechanism highlights that ambiguity aversion acts as a disciplining force: it tightens effective risk limits through endogenous belief distortion, even though the underlying contract remains unchanged. This channel differs from the pure effect of risk aversion documented in Figure~\ref{fig3}, which operates solely through utility curvature.

Figure \ref{fig:pit} plots the optimal investment ratio (vertical axis, $\pi_t^*/W^*_t$) as a function of current wealth (horizontal axis, $W^*_t$) for alternative ambiguity attitudes. The fraction invested in the risky asset declines with wealth regardless of the attitude. This downward slope reflects the concavified objective: as wealth rises, the marginal value of the additional upside diminishes, and the manager behaves more cautiously despite the convex contract. In addition, higher ambiguity aversion leads to a lower allocation to the risky asset. To illustrate, in Figure~\ref{fig:pit} the allocation corresponding to $\mathrm{RAA}=0.01$ (dashed, red) lies uniformly above that for $\mathrm{RAA}=0.02$ (dotted, green), which in turn lies above that for $\mathrm{RAA}=0.3$ (solid, blue) and converges rapidly to zero. For even higher values of the relative ambiguity aversion ($\mathrm{RAA}>0.3$), the risky share is always substantially lower and converges rapidly toward zero. Quantitatively, moving from low to high ambiguity aversion roughly halves the risky share at intermediate wealth levels. The explanation is that an increase in wealth reduces marginal utility, and ambiguity aversion shifts weight toward adverse scenarios. Both these effects discourage risky exposure. Thus, ambiguity aversion not only distorts beliefs but also induces more conservative dynamic trading, mitigating the risk-shifting incentives generated by the option-based contract.
\begin{table}[htbp]
	\centering
    \caption{Optimal $q^*$ for different AAA levels in the exponential-power case (initial wealth $w=10$, $\mathrm{RRA = 0.5}$)}
    \label{tab2}
    \begin{tabular}{ccccccccccc}
		\toprule
		AAA ($\gamma$) & 0.01&0.1 & 0.5 & 1 & 2 &3& 4 & 8 & 12 & 15 \\ 
		\midrule
		$q^*$ & 0.799 & 0.791&0.751 & 0.694 & 0.563 & 0.423& 0.294 & 0.108 & 0.094 & 0.090 \\
		\bottomrule
	\end{tabular}
\end{table}

Now we repeat the analysis for the exponential-power case of Proposition \ref{prop4.2}. We start with the probability distortion. Solving~\eqref{eqexp} yields $q^*=\mathcal{Q}^*(Z=z_2)$ and $1-q^*=\mathcal{Q}^*(Z=z_1)$ for different absolute ambiguity aversion (AAA) levels $\gamma$ (Table~\ref{tab2}). As $\gamma$ increases, $q^*$ decreases, meaning that the manager assigns increasingly higher probability to the worst state. Numerically, $q^*$ declines from $0.799$ to $0.090$ as $\gamma$ rises from $0.01$ to $15$, more than tripling the probability mass of the bad state. 
\begin{figure}[htbp]
	\centering
	\begin{minipage}[t]{0.47\textwidth}
			\includegraphics[width=\textwidth]{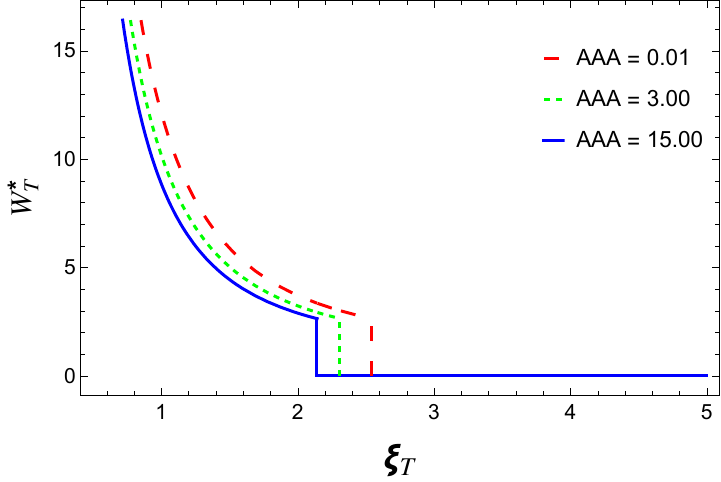}
			\caption{Optimal terminal wealth $W_T^*$ as a function of $\xi_T$ for three ambiguity aversion levels ($\mathrm{AAA=\gamma}$) in exponential-power case (initial wealth $w=10$, $\mathrm{RRA = 0.5}$).}
		\label{fig8}	
		\end{minipage}
	\hfill
    \begin{minipage}[t]{0.47\textwidth}
		\includegraphics[width=\textwidth]{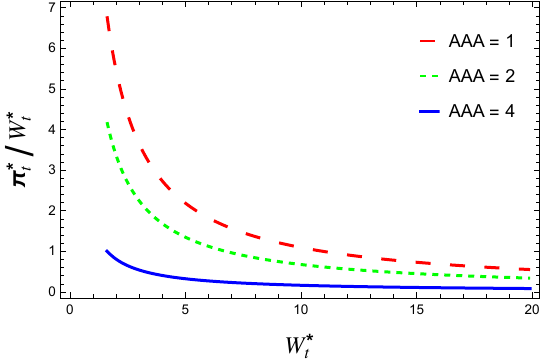}
		\caption{Optimal risky proportion vs.\ wealth $W^*_t$ for different ambiguity aversion levels ($\mathrm{AAA=\gamma}$) in the exponential-power case (initial wealth $w=10$, $\mathrm{RRA = 0.5}$, $t=T-1$).}
		\label{fig:piq}
	\end{minipage}	
	\end{figure}

 Similar to the power--power specification, Figure~\ref{fig8} plots optimal terminal wealth (vertical axis, $W_T^\ast$) against the state-price density (horizontal axis, $\xi_T$) for three levels of ambiguity aversion, $\mathrm{AAA}=0.01$ (dashed, red),  $\mathrm{AAA}=3$ (dotted, green) and $\mathrm{AAA}=15$ (solid, blue). The qualitative patterns resemble those observed in the power--power case. In addition, Figure~\ref{fig:piq} displays the optimal investment ratio (vertical axis, $\pi_t^\ast / W_t^\ast$) as a function of current wealth (horizontal axis, $W_t^\ast$) for three levels of ambiguity aversion: $\mathrm{AAA}=1$ (dashed, red), $\mathrm{AAA}=2$ (dotted, green) and $\mathrm{AAA}=4$ (solid, blue). We also find that the fraction invested in the risky asset declines with wealth across all ambiguity attitudes. Moreover, ambiguity aversion shifts probability weight toward adverse scenarios, thereby reducing risk exposure. 
These patterns are consistent with those obtained under the power--power specification, underscoring the consistency of the economic implications.

\section{Conclusion} \label{sec_conclusion} 
We have proposed a general framework for non-concave dynamic portfolio optimization under smooth ambiguity with Bayesian learning. 
To address the inherent dynamic inconsistency of smooth ambiguity preferences, we have employed the robust representation of quasiconcave functionals. This has allowed us to reformulate the original certainty-equivalent problem into an optimization problem with a penalty on the belief about the uncertain expected return. Under suitable specifications of $\phi$, the problem reduces to an ambiguity-neutral optimization with distorted priors, naturally interpreted as a Bayesian adaptive problem. Priors are updated via nonlinear filtering, and in combination with the martingale approach and concavification principle, this structure yields semi-closed-form optimal strategies.  

We have then applied this framework to option-based compensation schemes in fund management to study how ambiguity aversion interacts with convex incentive contracts. Our analysis shows that ambiguity aversion endogenously tilts beliefs toward adverse states, shrinks in-the-money regions for convex contract components, and lowers risky exposure, thereby mitigating risk-shifting incentives relative to the ambiguity-neutral benchmark. Thus, effectively, ambiguity aversion functions as an endogenous risk constraint in convex compensation environments.

Because the argument is developed within a general option-style payoff framework, the implications carry over to other payoff-engineering contexts, including guarantees, drawdowns, and insurance products with embedded options such as caps, floors, surrender features, and guaranteed minimum benefits (e.g., \cite{bacinello2011variable}).


This paper opens several avenues for future research. A first technical direction is to explore further and compare the robust representation of quasiconcave functionals with the equilibrium approaches studied in the time-inconsistency literature. For example, it would be interesting to compare our investment strategies with weak/strong/regular equilibria as in \cite{huang2021strong} and \cite{he2021equilibrium}. A second extension is to incorporate more realistic market features, such as portfolio constraints (\citep{dai2022nonconcave}) or transaction costs (\citep{qian2023non}). Finally, it would be valuable to investigate other non-concave problems beyond the option-based compensation considered here and to examine how ambiguity aversion shapes outcomes. Examples include goal-reaching problems \citep{capponi2024continuous}, S-shaped utility from prospect theory \citep{tse2023portfolio}, and convex performance fee schemes \citep{he2018profit}. We leave these questions for future work.

\clearpage

\section*{Acknowledgments}
An Chen and Shihao Zhu gratefully acknowledge financial support from the Deutsche Forschungsgemeinschaft (DFG), Project-ID 509303834, FOR 5583 ``Asset Allocation and Asset Pricing under Regulatory Uncertainty".

\bibliographystyle{abbrvnat}
\bibliography{KMM}

\clearpage
\appendix
\setcounter{subsection}{0}
\renewcommand\thesubsection{A.\arabic{subsection}}
\setcounter{equation}{0}
\renewcommand\theequation{A.\arabic{equation}}
\renewcommand{\thelemma}{\Alph{section}.\arabic{lemma}}

\renewcommand{\thetable}{A.\arabic{table}} 
\setcounter{table}{0}

 \section{Notation and Auxiliary Lemmata}




\begin{table}[htbp]
\centering
\caption{Notation for probability and expectation}
\label{notation}
\begin{tabular}{ll}
\toprule
\textbf{Symbol} & \textbf{Meaning} \\
\midrule
$\mathbb P \ (\mathbb E^{\mathbb P})$ & Generic probability measure (and expectation)   \\
$\mathbb P^Z \ (\mathbb E^{\mathbb P^Z})$  &Model (and expectation) \\
$\mathbb{P}_0$ & Reference probability measure, dominating all  $\mathbb P^Z$ \\
$\mathfrak{P}$ & Set of models \\
$\mathcal{P} \ (\mathbb E^{\mathcal P}) $ & Prior on $Z$ (and expectation) \\
$\mathcal{Q} \ (\mathbb E^{\mathcal Q})   $ & Distorted Prior on $Z$ (and expectation) \\
$\mathcal M(\mathcal S)$ & Set of priors on the parameter space $\mathcal S$\\
$\mathbb P^{\mathcal Q} \ (\mathbb E^{\mathbb P^{\mathcal Q}}) $ & Product measure of $\mathbb P^Z$ and $\mathcal Q$ (and expectation) \\
$\widehat{\mathbb P}  \ (\mathbb E^{\widehat{\mathbb P}}) $ & Equivalent measure of $\mathbb P$ in filtering  (and expectation)  \\
\bottomrule
\end{tabular}
\vspace{3mm}

\end{table}

\begin{lemma}\label{minimax}
Let $M$ and $N$ be convex spaces, one of which is compact, and $f$ a function on $M \times N$, quasi-concave-convex and upper semi-continuous-lower semi-continuous. Then
\[
\sup_x \inf_y f(x,y)=\inf_y \sup_x f(x,y).	
\]
\end{lemma}
\begin{proof}{Proof}
The above lemma can be found in  Corollary 3.3 of \cite{sion1958general}.
\end{proof}

\medskip
\begin{lemma}[Convex conjugate for payoff $g(x)$]\label{convex} 
Define $\mathcal{L}(y,x):=u(g(x))-yx, y >0$, $I(x):=x^{\frac{1}{\alpha-1}}$ and $h(x):=\frac{I(x/\delta)-C}{\delta}+K$. Then 
\begin{equation}\label{eq-a1}
\mathcal{X}(y):= \argsup_{x\geq 0} \mathcal{L}(y,x)
\\
=\left\{
\begin{aligned}
h(y), \quad &0<y<\widehat{y}, \\
0, \quad & y \geq \widehat{y},
\end{aligned}
\right.
\end{equation}
where the concavification point $\widehat{y}$ is defined as the unique root $y\in (0,  \delta C^{\alpha-1})$ of 
\[
u(g(h(y)))- u(C)=y h(y).
\]
\end{lemma}
\begin{proof}{Proof}
Note first that $\mathcal{L}$ is continuous in $x$. Except at $K$ its derivative with respect to $x$ exists and is given by
\begin{equation}\label{eqa-2}
\mathcal L_x(y,x)=
\begin{cases}
-y, &\quad 0 \leq x <K,\\
\delta u'(\delta( x-K)+C)-y, & \quad x>K.
\end{cases}
\end{equation}
Solving the first-order condition $\mathcal L_x(y,x) \overset{!}{=} 0  $ for $x$, we obtain a local extrema $h(y)$ (if $h(y)$ is larger than $K$). Further, the edge points $0$ and $K$ are candidates for the maximizer $\mathcal{X}(y)$. At the edge point  $K$, we obtain 
\[
\mathcal L_x(y,K+)=\delta u'(C)-y=\delta C^{\alpha-1}-y.
\]
This determines two cases that are used to determine $\mathcal{X}(y)$: 

\textbf{Case 1:} $y\geq \delta C^{\alpha-1}$. Using $\mathcal{L}_x$ in (\ref{eqa-2}), we obtain the monotonicity pattern in $x$:  $\mathcal{L}$ is decreasing on $(0,+\infty)$. As $\mathcal{ L}$ is continuous in $x$, we obtain $\mathcal{X}(y)=0$ when $y\geq \delta C^{\alpha-1}$.

\textbf{Case 2:} $0 <  y<\delta C^{\alpha-1}$.  We observe that the monotonicity pattern in $x$, i.e.,  $\mathcal{L}$ is decreasing on $(0,K)$, increasing on $(K,h(y))$, decreasing on $(h(y),+\infty)$. This leaves two candidates for the maximizer $\mathcal{X}(y)$, namely $0$ and $h(y)$. To compare $\mathcal{L}(y,0)$ and $\mathcal{L}(y,h(y))$, we observe that 
\[
\mathcal{L}(y,h(y))-\mathcal{L}(y,0)
=u(g(h(y)))-u(C)-yh(y)
  \]
due to $h(y)>K$. Then we define the function 
\[
\Delta \mathcal{L}(y):= u(g(h(y)))-u(C)-yh(y)
\]
and we compute 
\begin{equation*}
    \begin{aligned}
        \frac{\partial \Delta \mathcal L(y)}{\partial y}&= \delta [u'(\delta(h(y)-K)+C)-y]h'(y)-h(y)\\
    &=-h(y)<0.    \end{aligned}
\end{equation*}
This implies that $\Delta\mathcal L(y)$ is strictly decreasing on $(0, \delta C^{\alpha-1})$. Moreover, we compute 
\begin{equation*}
\Delta \mathcal L(0)=\begin{cases}
      u(\infty)-u(c)=+\infty>0, \ &\text{if} \ 0<\alpha<1, \\
	 -\frac{C^\alpha}{\alpha}>0,  \ &\text{if} \ \alpha<0, 	 
   \end{cases}
\end{equation*}
and $\Delta \mathcal L(\delta C^{\alpha-1})=-\delta C^\alpha <0 $. Therefore, there exists an unique root $\widehat{y} \in (0,\delta C^{\alpha-1})$ of $u(g(h(y)))- u(C)=y h(y).$ Furthermore, when $y\in (0,\widehat{y})$, $\Delta \mathcal L>0$, thus $h(y)$ is the maximizer; when  $y\in (\widehat{y},+\infty), $ $0$ is the maximizer, which finishes the proof.

\end{proof}

\medskip
\subsection{Filtering Arguments}\label{filtering}

In this section, we reformulate the Bayesian adaptive problem by applying filtering theory, thereby modeling the drift as a process adapted to the observation filtration. This transformation lays the foundation for the proof of Theorem \ref{the4.1}.

Define the market price of risk $\Theta=(Z-r)/\sigma$ with $\theta=(z-r)/\sigma$.
The setting concerns a financial market in which the {stock price process} $S$ is perfectly observable in continuous time, whereas the {drift parameter} $\Theta$ that drives the expected return of the stock is {not} directly observable.  
Consequently, the investor can only learn about $\Theta$ by observing the evolution of $S$ and updating her belief on $\Theta$ via Bayesian filtering.

 Note further that observing the stock price is equivalent to monitoring the process \( Y=(Y_t)_{t\in[0,T]} \) as \( Y_t := W_t+ \Theta t \).  
The price dynamic of the risky asset \( S  \) in (\ref{stock}) can be written as
\[
dS_t = S_t\!\left(r \, dt + \sigma \, dY_t\right), \quad t \ge 0,
\]
where $Y_t$ is a $\mathbb{P}$-Brownian motion with drift $\Theta$.  
Let $\mathbb{F}^Y:=(\mathcal{F}^Y_t )_{t\in[0,T]}$, where  
$\mathcal{F}^Y_t := \sigma(Y_s,0\le s \le t)$ is the filtration generated by $Y$---equivalent to the filtration generated by $S$.

As time progresses, the investor draws her inferences about $\Theta$ and updates
her prior via $\widehat{\theta}:=(\widehat{\theta}_t)_{t\in [0,T]}$:
\[
\widehat{\theta}_t := \mathbb{E}^{\mathcal{Q}}[\Theta\mid \mathcal{F}^Y_t].
\quad t \in [0,T],
\]
Note that condition (\ref{prior-integrability}) ensures that $(\widehat{\theta}_t)_{t\in[0,T]}$ is well-defined.


 Let the process $B^Y=(B^Y_t)_{t\in[0,T]}$ be given by
\begin{equation}\label{eqa-3-1}
	dB_t^Y = dB_t + (\Theta - \widehat{\theta}_t)\,dt.
	\end{equation}
The process $(B_t^Y)_{t\in[0,T]}$ is called the innovation process in filtering theory. It is well known that $B^Y$ is an $(\mathbb{F}^Y,\mathbb{P})$ standard Brownian motion
(see, for example, Proposition 2 in \cite{bismuth2019portfolio}).  
Moreover, by Theorem 1 in \cite{bismuth2019portfolio} (see also \cite{karatzas2001bayesian}, \cite{rieder2005portfolio}), we have
\begin{equation} \label{eq_thetahat} 
\widehat{\theta}_t = \frac{F_y(t,Y_t;\mathcal Q)}{F(t,Y_t; \mathcal Q)}, \quad t\in[0,T],
\end{equation}
where
\begin{equation}\label{eqb2}
F(t,y; \mathcal Q):= \int_{\mathcal{S}}\exp\!\bigl\{\theta y-\tfrac{1}{2}\theta^2 t\bigr\}\,\mathcal{Q}(d\theta).
\end{equation}

With this conditional expectation---i.e. based on the investor's Bayesian update $\widehat{\theta}$---the stock-price evolution becomes
\[
dS_t = S_t\!\left(\bigl(r + \sigma \widehat{\theta}_t\bigr)\,dt + \sigma \, dB_t^Y\right),
\]
where $B_t^Y$ is given in (\ref{eqa-3-1}).

\noindent Note that due to the relation 
\[
Y_t={B}^Y_t+\int^t_0\widehat{\theta}_sds,
\]
we can define a new probability measure $\widehat{\mathbb{P}}$ for each fixed $T$, under which $Y_t$ becomes a standard $(\mathbb{F}^Y,\widehat{\mathbb{P}})$ Brownian motion for $0\leq t \leq T$:
\[
\Lambda_T:=\frac{d\widehat{\mathbb{P}}}{d \mathbb{P}}\bigg|_{\mathcal F^Y_ T}= \exp\bigg(-\int^T_0 \widehat{\theta}_t d{B}^Y_t-\frac{1}{2}\int^T_0\widehat{\theta}_t^2dt  \bigg).
\]

{\raggedleft{An}} application of It\^o's rule to the process $\Lambda_t$ gives
\[
d{\Lambda}_t=-{\Lambda}_t \widehat{\theta}_tdB^Y_t, \quad {\Lambda}_0=1.
\]
Now we can rewrite (\ref{wealth}) as 
\begin{equation}\label{eq4-6-1-1}
\begin{aligned}
    dW_t^\pi&=(rW^\pi_t+\sigma \pi_t \widehat{\theta}_t)dt+\pi_t \sigma dB_t^Y \\&=rW^\pi_t dt+\pi_t\sigma \; d Y_t,
    \end{aligned}
\end{equation}
where $\widehat{\theta}_t$ is given in \eqref{eq_thetahat}.  Further, we have 
\[
d ({\Lambda}_t e^{-rt}W_t^\pi )= e^{-rt}{\Lambda}_t[\sigma\pi_t -W_t^\pi\widehat{\theta}_t]dB^Y_t.
\]
This shows that, on a given finite time horizon $[0,T]$, the process $(e^{-rt}{\Lambda}_t W_t^\pi)_{t\in[0,T]}$ is a nonnegative ($\mathbb{F}^Y,\mathbb{P}$)-local martingale, hence also a supermartingale; in particular, for all $\pi \in \mathcal{A}(w)$, 
\begin{equation}\label{3-2}
\mathbb{E}^{\mathbb{P}}[e^{-rT}{\Lambda}_T W_T^\pi]=\mathbb{E}^{\widehat{\mathbb{P}}}[e^{-rT} W_T^\pi]   \leq w,
\end{equation}
where $\mathbb E^{\mathbb P}$ and ${\mathbb{E}^{\widehat{\mathbb{P}}}}$ are  expectations with respect to  $\mathbb P$ and $\widehat{\mathbb{P}}$, respectively.

Moreover, we state  the following lemma for later use. 
\begin{lemma}\label{lambda}
We can express $\Lambda_t$ as a function of the observational magnitude $Y_t$:
\[
\Lambda_t=\frac{1}{F(t,Y_t;\mathcal Q)}, \quad t\geq 0,
\]
where $F(t,y;\mathcal Q)$ is given in (\ref{eqb2}).
\end{lemma}
\begin{proof}{Proof}
The proof can be found in Lemma 2.1 in \cite{bauerle2019optimal}.
\end{proof}

\section{Technical Proofs}
\renewcommand\theequation{B.\arabic{equation}}
\renewcommand\thesubsection{B.\arabic{subsection}}

\subsection{Proof of Theorem \ref{theorem3.1}}
\begin{proof}{Proof}\label{proofthe3.1}
We apply Lemma \ref{minimax} to the function $L: \mathcal{A}(w) \times \mathcal{M(}\mathcal{S}) \to \mathbb{R}$ defined by 
\[
L(\pi,\mathcal{Q}):=R(\mathcal{Q}, \mathbb{E}^{\mathbb{P}^\mathcal{Q}}[U(W_T^\pi)] ).
\]
The sets $\mathcal{A}(w)$ and $\mathcal{M}(\mathcal{S})$ are clearly convex and $\mathcal{A}$ is a compact set.   	Point (2) of Proposition \ref{prop3.1} implies $L(\pi,\cdot)$ is quasi-convex for any $\pi \in \mathcal{A}$. Then we notice that $\pi \to W^\pi_T$ is linear. Moreover $U$ is increasing and $R$ is nondecreasing in the second argument, so that $L(\cdot,\mathcal{Q})$ is quasi-concave for any $\mathcal{Q}$. Thus, $L$ is quasi-concave-convex, i.e., $L$ is quasi-concave in $\pi$ and quasi-convex in $\mathcal{Q}$.

 Moreover, the function $L(\mathcal{Q},\cdot)$ is upper semicontinuous by point (1) of Proposition \ref{prop3.1} and since
\[
\pi \mapsto \mathbb{E}^{\mathbb{P}^\mathcal{Q}}[U(W_T^\pi)]  
\]
is continuous.  It remains to show that $L(\pi,\cdot)$ is lower semi-continuous for all $\pi\in \mathcal{A}$. Assume first that we are in the context of Assumption \ref{assume} (1). Then, for any  sequence $(\mathcal{Q}^n)$ converging to a limit $\mathcal{Q} \in \mathcal{M}(\mathcal{S})$ and for any index $k \in \mathbb{N}$, 
\begin{equation*}
    \begin{aligned}
        \lim_n \inf L(\pi, \mathcal{Q}^n)&=\lim_n \inf_{m\geq n}R(\mathcal{Q}^m, \mathbb{E}^{\mathbb{P}^{\mathcal{Q}^m}}[U(W_T^\pi)] )	\\
&\geq \lim_n \inf_{m\geq n}R(\mathcal{Q}^m, \inf_{j\geq k}\mathbb{E}^{\mathbb{P}^{\mathcal{Q}^j}}[U(W_T^\pi)] )		\\
&\geq R(\mathcal{Q}, \inf_{j\geq k}\mathbb{E}^{\mathbb{P}^{\mathcal{Q}^j}}[U(W_T^\pi)] ),	
\end{aligned}
\end{equation*}
where the first inequality holds because $R$ is nondecreasing in $s$. Now, making $k \to +\infty$, by continuity we obtain 
\[
\lim_n \inf L(\pi,\mathcal{Q}^n) \geq L(\pi,\mathcal{Q}).
\]
Under Assumption \ref{assume} (2) a similar argument can be used since 
\[
\mathbb{E}^{\mathbb{P}^{\mathcal{Q}}}[U(W_T^\pi)]>0. 		
\]
Therefore, the infimum and the supremum in Lemma \ref{minimax} are attained and there exists a saddle point.
\end{proof}

\subsection{Proof of Theorem \ref{the4.1}}\label{proofthe4.1}

\begin{proof}{Proof}
Now we use martingale method, to maximize the expected utility 
\[
V(w;\mathcal Q)=	\sup_{\pi \in \mathcal{A}} \mathbb{E}^{\mathbb{P}}[u(g(W_T^\pi))]
\]
subject to the constraint (\ref{3-2}) where $W_T^\pi$ is given in (\ref{eq4-6-1-1}), that is, we have the following static optimization problem:
\begin{equation}\label{eq3-8}
W_T^*=\argmax_{W^{\pi}_T\geq 0}  \mathbb{E}^{\mathbb{P}}[u(g(W_T^\pi))]	
\end{equation}			
such that $\mathbb{E}^{\mathbb{P}}[\xi_T W_T^\pi]\leq w$ with $\xi_T=e^{-rT}{\Lambda}_T $.

To solve (\ref{eq3-8}), consider the Lagrangian  problem 
\begin{equation*}
\sup_{W_T^\pi\geq 0} \mathbb{E}^{\mathbb{P}}[u(g(W_T^\pi))+\kappa (w-\xi_T W_T^\pi ) ]	
\end{equation*}
for a multiplier $\kappa\geq0$.  Define $J(y):= \sup_{x>0} \{u(g(x))-yx \}$. The existence of $J$ is ensured by the growth condition in Definition \ref{def:utility}. The maximizer in $J(y)$ is  given  by $\mathcal{X}$ (cf.\ (\ref{eq-a1})), i.e.,
\[
\mathcal X(y)=h(y) \cdot \mathbbm{1}_{\{0<y<\widehat{y}\}}.	
\]
This leads to $W_T^*(\kappa,\xi_T)=\mathcal{X}(\kappa \xi_T)$ as a candidate for optimal terminal wealth. For any admissible terminal wealth $W_T^{\pi}$, we get 
\begin{equation*}
    \begin{aligned}
        \mathbb{E}^{\mathbb{P}}[u(g(W_T^\pi))] &\leq \mathbb{E}^{\mathbb{P}}[u(g(W_T^\pi))+\kappa (w-\xi_T W_T^\pi ) ]	\\
&\leq \mathbb{E}^{\mathbb{P}}[J(\kappa \xi_T)] +\kappa w	\\
&=\mathbb{E}^{\mathbb{P}}[u(g(\mathcal{X}(\kappa \xi_T))) ],  
\end{aligned}
\end{equation*}
i.e., optimal terminal wealth is indeed given by $W_T^*(\kappa,\xi_T)=\mathcal{X}(\kappa \xi_T)$. In the last equality, we have used that the budget constraint is binding for an optimal solution to  (\ref{eq3-8}).


The proof of existence of the Lagrangian multiplier $\kappa^*$ is as follows. Fix $\xi_T>0$. Define 
\[
\chi(\kappa):=	\mathbb{E}^{\mathbb{P}}[\xi_T W_T^*(\kappa,\xi_T) ]=\mathbb{E}^{\mathbb{P}}[\xi_T \mathcal X(\kappa\xi_T) ].
\]
We first show that $\chi (\kappa)$ is monotone decreasing in $\kappa$. By Lemma \ref{lambda}, we observe that 
\begin{equation*}
    \begin{aligned}
        \chi (\kappa)={\mathbb{E}^{\widehat{\mathbb{P}}}}[e^{-rT}\mathcal X(\kappa\xi_T)]={\mathbb{E}^{\widehat{\mathbb{P}}}}[e^{-rT}\mathcal X(\kappa e^{-rT}/F(T,Y_T))].    \end{aligned}
\end{equation*}
Clearly, $\chi(\kappa)$ is continuous and decreasing. Furthermore, $\lim_{\kappa \to 0}\chi(\kappa) = +\infty$, and $\lim_{\kappa\to \infty}\chi(\kappa)=0$. Therefore, there exists a unique $\kappa^*_w>0$ such that $\chi(\kappa^*_w)=w$.	

Next we show the optimal investment strategy. By martingale representation property of the Brownian filtration (e.g.,\ Theorem 3.4.2 in \cite{karatzas2012brownian}), we obtain 
\begin{equation}\label{3-8}
    \begin{aligned}
        e^{-rt}\widehat{W}_t:&=e^{-rT} {\mathbb{E}^{\widehat{\mathbb{P}}}}[\mathcal X(\kappa^*_w \xi_T)|\mathcal{F}_t^Y] 
\\&= w+ \int^t_0 e^{-rs} \widehat{\pi}_s \sigma dY_s, \quad t \in [0,T]    \end{aligned}
\end{equation}
for some $\mathbb{F}^Y$-progressively measurable process $\widehat{\pi}:[0,T] \times \Omega \to \mathbb{R}$ that satisfies $\int^T_0 e^{-2rs} ||\widehat{\pi}_t||^2dt<\infty$ a.s. (with respect to both $\mathbb{P}$ and $\widehat{\mathbb{P}}$). Furthermore, by comparing  (\ref{eq4-6-1-1}) and (\ref{3-8}), we have 
\begin{equation}\label{3-9}
\begin{aligned}
    W^{\widehat{\pi}}_t=\widehat{W}_t&= e^{-r(T-t)} \mathbb{E}^{\widehat{\mathbb{P}}}[\mathcal X(\kappa^*_w \xi_T)|\mathcal{F}_t^Y]\\&= \mathcal{Y}(T-t,Y_t), \quad t \in [0,T],\end{aligned}
\end{equation}
where 
\begin{equation}\label{3-10}
\mathcal{Y}(s,y)
:=
\begin{cases}
\displaystyle
e^{-rs}
\int_{\mathbb{R}}
\mathcal X\!\left(
\frac{\kappa^*_w e^{-rT}}{F(T,y+z)}
\right)
\varphi_s(z)\,dz,
& 0 < s \le T, \\[0.5em]
\displaystyle
\mathcal X\!\left(
\frac{\kappa^*_w e^{-rT}}{F(T,y)}
\right),
& s = 0 .
\end{cases}
\end{equation}
Here $\varphi_s(z):=(2\pi s)^{-1/2} e^{-z^2/(2s)} $ is the Gaussian density function.
Together with (\ref{3-8}), the equations (\ref{3-9}) and (\ref{3-10}) lead to the expression 
\begin{equation*}
\pi^*_t=\widehat{\pi}_t=\frac{1}{\sigma} \nabla \mathcal{Y}(T-t,Y_t), \quad 0\leq t<T.
\end{equation*}
Finally, in conjunction with (\ref{3-10}), we have the gradient 
{
\begin{equation*}
\nabla \mathcal{Y}(s,y)
= -\kappa^*_w e^{-r(T+s)}
   \int_{\mathbb{R}}
   \frac{G(T,y+z)}{F(T,y+z)} 
   \mathcal X'\!\left(
     \frac{\kappa^*_w e^{-rT}}{F(T,y+z)}
   \right)
   \varphi_s(z)\,dz .
\end{equation*}
}
where $G(t,y):=\frac{\nabla F}{F}(t,y)$. 
Then the value function of problem (\ref{eq4-7}) is 
\begin{equation*}
    \begin{aligned}
        V(w;\mathcal Q) 
&= \mathbb{E}^{\mathbb{P}}
   \bigl[u\bigl(g(W_T^*)\bigr)\bigr]\nonumber\\
&= \mathbb{E}^{\mathbb{P}}
   \Bigl[u\bigl(g(\mathcal X(\kappa^*_w e^{-rT}/F(T,Y_T)))\bigr)\Bigr]\nonumber\\
&= \mathbb{E}^{\widehat{\mathbb{P}}}
   \Bigl[
     F(T,Y_T)\,
     u\bigl(g(\mathcal X(e^{-rT}\kappa^*_w/F(T,Y_T)))\bigr)
   \Bigr]\nonumber\\
&= \int_{\mathbb{R}}
     F(T,z)\,
     u\!\left(
       g\!\left(
         \mathcal X\!\left(
           \frac{e^{-rT}\kappa^*_w}{F(T,z)}
         \right)
       \right)
     \right)
     \varphi_T(z)\,dz .   \end{aligned}
\end{equation*}

\end{proof}

\section{Factorial Design}\label{design}
\renewcommand\thetable{C.\arabic{table}}
\renewcommand\thesubsection{C.\arabic{subsection}}

\begin{table}[htbp]
\centering
\caption{Full $2^3$ factorial design}
\label{table5}
\begin{tabular}{cccc}
\hline
Prior (A) & RRA (B) & RAA (C) & $q^*$ \\
\hline
Low & Low & Low & 0.392 \\
High & Low & Low & 0.663 \\
Low & High & Low & 0.413 \\
High & High & Low & 0.733 \\
Low & Low & High & 0.079 \\
High & Low & High & 0.103 \\
Low & High & High & 0.104 \\
High & High & High & 0.127 \\
\hline
\end{tabular}
\end{table}

We implement a $2^3$ design with $x_A, x_B, x_C \in \{-1,+1\}$ 
denoting the baseline prior, RRA, and RAA. 
Effects are obtained from the saturated linear representation
\begin{equation*}
q^* = \beta_0 + \beta_A x_A + \beta_B x_B + \beta_C x_C
+ \beta_{AB} x_A x_B + \beta_{AC} x_A x_C
+ \beta_{BC} x_B x_C + \beta_{ABC} x_A x_B x_C.
\end{equation*}

Under $\pm1$ coding, each effect equals the difference in conditional means. Table~\ref{table5} reports the eight solutions.


\end{document}